\documentclass[11pt]{amsart}

\usepackage[a4paper,top=2cm,bottom=2cm,left=2cm,right=2cm]{geometry}

\usepackage[T1]{fontenc}
\usepackage{amsmath,amssymb,amsthm}
\usepackage{graphicx}
\usepackage{subcaption}
\usepackage{float}
\usepackage{hyperref}
\usepackage{placeins}
\usepackage{amsmath}
\usepackage{float}

\newcommand{\sech}{\mathop{\rm sech}\nolimits}

\DeclareMathOperator{\Ran}{Ran}
\newtheorem{lemma}{Lemma}[section]

\newtheorem{theorem}{Theorem}[section]
\newtheorem{proposition}{Proposition}[section]

\newtheorem{definition}{Definition}[section]
\newtheorem{remark}{Remark}[section]

\theoremstyle{definition}

\numberwithin{equation}{section}

\title[Traveling waves in a regularized sine--Gordon equation]
{Spectral stability and slow--fast structure of traveling waves in a regularized sine--Gordon equation}

\author{Vassilios M Rothos}
\address{School of Mechanical Engineering, Faculty of Engineering\\
Aristotle University of Thessaloniki,\\
Thessaloniki 54124, Greece}
\email{rothos@auth.gr}

\thanks{Accepted for publication in Wave Motion. 
The research was funded by Aristotle University of Thessaloniki (AUTH) Research Council grants number 73191, 73699, 11682.}

\date{\today}

\dedicatory{}

\begin{document}
\begin{abstract}
We investigate the dynamics and spectral stability of traveling kink and antikink solutions in a dissipative sine--Gordon equation with two distinct fourth--order regularization mechanisms: a mixed space--time (inertial) term and a purely spatial (elliptic) term. The model includes damping, bias forcing, and higher--order dissipative effects, and is motivated by refined descriptions of fluxon dynamics in long Josephson junctions. Using a collective--coordinate reduction, we derive a Melnikov--type condition for speed selection, yielding explicit predictions for asymptotic propagation speeds, which are validated by direct numerical simulations of the full partial differential equation. Spectral stability is analyzed using Evans function techniques adapted to the singular slow--fast structure induced by the regularization. By formulating the linearized problem on a consistent--splitting domain, we show that no additional point spectrum bifurcates from the origin. Near the edges of the essential spectrum, a square--root transformation is used to resolve branch singularities and establish analyticity in a lifted spectral variable. Numerical Evans function computations near $\lambda=0$ and near the essential spectrum edges confirm the analytical results, indicating absence of unstable eigenvalues for both kink and antikink solutions.

\end{abstract}
\maketitle
\section{Introduction}
\label{sec1}

The sine--Gordon equation is a fundamental model for nonlinear wave propagation
in a wide range of physical systems, most notably in the theory of long Josephson
junctions, where traveling wave solutions represent propagating magnetic flux
quanta (fluxons); see, for example,
\cite{BaronePaterno,Scott1978,Trullinger}.
In idealized settings, the classical sine--Gordon equation
admits an integrable Hamiltonian structure and exact traveling kink solutions.
In realistic physical regimes, however, dissipation, external forcing, and
higher--order regularization effects are unavoidable and can substantially
influence both the dynamics and stability of such waves
\cite{KivsharMalomed,McLaughlinScott}.

In this work we consider a dissipative and regularized sine--Gordon equation of
the form
\begin{equation}
u_{tt}-u_{xx}+\sin u
=
\varepsilon\Big(
\gamma-\alpha u_t+\beta u_{xxt}
+\delta u_{xxtt}-d u_{xxxx}
\Big),
\qquad 0<\varepsilon\ll1,
\label{eq:SG_general_intr}
\end{equation}
where $\alpha>0$ denotes linear damping, $\gamma\in\mathbb{R}$ is a constant bias,
and $\beta\ge0$ accounts for higher--order diffusive dissipation.
The parameters $\delta\ge0$ and $d\ge0$ represent two distinct fourth--order
regularization mechanisms: a mixed space--time (inertial) correction
$\delta u_{xxtt}$ and a purely spatial (elliptic) correction $-d u_{xxxx}$. Models of this type arise in refined descriptions of long Josephson junctions and
related superconducting devices; see, for example,
\cite{BaronePaterno,Scott1978,McLaughlinScott,KivsharMalomed,DDvGV}.

Regularized sine--Gordon models of this type arise naturally in refined circuit
descriptions of extended Josephson junctions and related superconducting devices
\cite{BaronePaterno,Goldobin}, where idealized second--order models fail to
capture inertial, capacitive, and electromagnetic relaxation effects at small
spatial and temporal scales.
From a modeling perspective, it is essential to distinguish between the roles of
the two regularization mechanisms present in \eqref{eq:SG_general_intr}.
The elliptic term $-d u_{xxxx}$ introduces spatial smoothing and modifies the
high--frequency behavior of the spectrum in a uniformly regular manner.
In contrast, the mixed space--time term $\delta u_{xxtt}$ constitutes a
\emph{singular} perturbation of the underlying wave equation, introducing
additional inertial modes and a pronounced slow--fast structure in the associated
spatial dynamics \cite{Fenichel,Jones1995}.

From a spectral point of view, these two regularization mechanisms have
fundamentally different consequences.
While the $d u_{xxxx}$ term leads to a higher--order but uniformly elliptic spatial
operator, the $\delta u_{xxtt}$ term increases the order of the spatial eigenvalue
problem in a singular fashion and alters the dispersion relation in a way that can
generate additional fast spatial modes.
A priori, the combined presence of these terms could modify the essential
spectrum, generate new point spectrum, or destabilize the neutral translational
eigenvalue at $\lambda=0$ associated with translation invariance.
Determining whether and how such effects occur is a central mathematical and
physical question in the spectral theory of nonlinear waves
\cite{Sandstede2002,PegoWeinstein,ZumbrunHoward}.

The purpose of this work is to provide a combined analytical and numerical
investigation of the dynamics and spectral stability of traveling kink and
antikink solutions of \eqref{eq:SG_general_intr}.
Our analysis addresses two complementary aspects.
First, we study the long--time dynamics of traveling waves using a
collective--coordinate reduction, leading to a Melnikov--type condition for speed
selection \cite{McLaughlinScott,Kapitula1999}.
The resulting Melnikov function predicts selected propagation speeds that depend
explicitly on the regularization parameters $\delta$ and $d$, reflecting the
distinct ways in which inertial and elliptic corrections influence the wave
dynamics.
Second, we analyze the spectral stability of the resulting traveling waves using
Evans function techniques adapted to the slow--fast structure induced by the
$\delta$--term \cite{AlexanderGardnerJones,KapitulaSandstede1998,KapitulaSandstede2002,SandstedeScheel}.

\subsection{Main results}

From a mathematical standpoint, the main contribution of this work is to show that,
despite the presence of two distinct fourth--order regularization mechanisms, the
local spectral stability properties of traveling waves remain governed by the
underlying slow dynamics.

More precisely, by formulating the linearized spectral problem as a first--order
Evans system on a consistent--splitting domain and exploiting the separation
between slow and fast spatial modes, we establish that neither the inertial
regularization controlled by $\delta$ nor the elliptic regularization controlled
by $d$ generates additional point spectrum near $\lambda=0$ or at the edges of the
essential spectrum.

In particular:
\begin{itemize}
\item the traveling kink and antikink solutions persist as smooth perturbations of
the classical sine--Gordon waves, with their speed selected by a Melnikov--type
condition depending explicitly on $\delta$ and $d$;

\item the spatial dynamics admit a uniform slow--fast decomposition, consistent
with geometric singular perturbation theory \cite{Fenichel,Jones1995}, with fast
modes remaining uniformly hyperbolic and decoupled from the spectral stability
problem;

\item the Evans function admits a corresponding factorization into slow and fast
components, in the sense of modern Evans function theory
\cite{AlexanderGardnerJones,KapitulaSandstede2002,SandstedeScheel}, where the fast
factor is analytic and nonvanishing in admissible spectral domains;

\item as a consequence, no additional eigenvalues bifurcate from $\lambda=0$ or
from the edges of the essential spectrum, and spectral stability is determined
entirely by the reduced slow dynamics.
\end{itemize}

Complementing the analytical results, we perform numerical simulations of the full
partial differential equation to validate the Melnikov speed selection mechanism
and to compute the Evans function near $\lambda=0$ and near the edges of the
essential spectrum. These computations confirm that variations in $\delta$ and $d$
lead to smooth and controlled changes in the selected wave speed, while preserving
spectral stability for sufficiently small $\varepsilon$.

\subsection{Organization of the paper}

The paper is organized as follows.
In Section~\ref{sec2} we derive the traveling wave equations and formulate the associated
spectral problem.
Sections~\ref{sec_Evans0},\ref{sec:Evans} introduce the Evans function framework and establish the slow--fast
decomposition of the spatial dynamics.
Sections~\ref{sec:local_spectral_structure}-\ref{sec:singular_limits} analyze the point spectrum near $\lambda=0$, while Section~\ref{sec:evans_near_0_and_edges_beta_delta_d} is devoted
to the structure of the essential spectrum and its edges.
Section~\ref{sec7} presents numerical simulations of kink and antikink dynamics together
with Evans function computations.
We conclude with a discussion of the implications of our results and directions
for future work.
\section{Travelling wave reduction and Melnikov analysis}
\label{sec2}

\subsection{Travelling wave reduction}
\label{sec2a}

We consider the perturbed sine--Gordon equation \eqref{eq:SG_general_intr} and seek travelling waves of the form $u(x,t)=U(\xi)$ with $\xi=x-ct$, where $c$ is the
wave speed. Substituting into \eqref{eq:SG_general_intr} yields the fourth--order ODE
\begin{equation}
(1-c^2)U''+\sin U
=
\varepsilon\Big(
\gamma+\alpha c U'
-\beta c U'''
+(\delta c^2-d)\,U''''
\Big).
\label{eq:TW_general}
\end{equation}
The highest derivative is multiplied by the small factor $\varepsilon(\delta c^2-d)$,
so \eqref{eq:TW_general} is singularly perturbed in the regime
\begin{equation}
\varepsilon|\delta c^2-d|\ll1,
\qquad \text{with } \delta c^2-d\neq 0,
\label{eq:nondegeneracy}
\end{equation}
where the nondegeneracy condition excludes the exceptional resonance at which the
fourth--order correction disappears in the travelling--wave reduction.

Introduce the variables
\[
U_1=U,\qquad U_2=U',\qquad U_3=U'',\qquad U_4=U'''.
\]
Then \eqref{eq:TW_general} can be written as the four--dimensional system
\begin{equation}
\begin{aligned}
U_1'&=U_2,\\
U_2'&=U_3,\\
U_3'&=U_4,\\
(1-c^2)U_3+\sin U_1&=\varepsilon\Big(\gamma+\alpha c U_2-\beta c U_4+(\delta c^2-d)\,U_4'\Big).
\end{aligned}
\label{eq:TW_system_raw}
\end{equation}
To expose the slow--fast structure, we isolate the term containing the highest derivative:
\begin{equation}
(\delta c^2-d)\,\varepsilon\,U_4'
=
(1-c^2)U_3+\sin U_1
-\varepsilon\big(\gamma+\alpha c U_2-\beta c U_4\big).
\label{eq:U4prime_isolated}
\end{equation}
When $\varepsilon\to0$, the left--hand side vanishes and \eqref{eq:U4prime_isolated} yields the
algebraic constraint
\begin{equation}
(1-c^2)U_3+\sin U_1=0.
\label{eq:constraint}
\end{equation}
This relation defines the critical manifold of the fast--slow system.

Introduce the slow spatial variable $z=\varepsilon \xi$, so that
\[
\frac{d}{d\xi}=\varepsilon\frac{d}{dz}.
\]
Equivalently, $\xi$ may be regarded as the fast scale and $z$ as the slow scale.
Setting $\varepsilon=0$ in \eqref{eq:TW_system_raw} gives the layer problem
\begin{equation}
\begin{aligned}
U_1'&=U_2,\\
U_2'&=U_3,\\
U_3'&=U_4,\\
0&=(1-c^2)U_3+\sin U_1,
\end{aligned}
\label{eq:layer_problem}
\end{equation}
where prime denotes differentiation with respect to the fast variable $\xi$.
The algebraic equation in \eqref{eq:layer_problem} defines the critical manifold
\begin{equation}
\mathcal M_0
=
\Big\{(U_1,U_2,U_3,U_4)\in\mathbb R^4:\ (1-c^2)U_3+\sin U_1=0\Big\},
\label{eq:M0_def}
\end{equation}
which can be written as the graph
\begin{equation}
U_3=h_0(U_1):=-\frac{\sin U_1}{1-c^2},
\qquad |c|\neq 1.
\label{eq:h0}
\end{equation}

To determine normal hyperbolicity, we rewrite \eqref{eq:U4prime_isolated} as
\begin{equation}
U_4'
=
\frac{(1-c^2)U_3+\sin U_1}{(\delta c^2-d)\varepsilon}
-\frac{\gamma+\alpha c U_2-\beta c U_4}{\delta c^2-d}.
\label{eq:fast_rhs}
\end{equation}
The presence of the factor $((\delta c^2-d)\varepsilon)^{-1}$ shows that deviations from the
constraint \eqref{eq:constraint} relax on the fast scale, provided $\delta c^2-d\neq0$.
Linearization of \eqref{eq:fast_rhs} in the transverse direction to $\mathcal M_0$ yields a
nonzero fast eigenvalue proportional to $((\delta c^2-d)\varepsilon)^{-1}$, which remains uniformly
away from the imaginary axis for $\varepsilon$ sufficiently small under \eqref{eq:nondegeneracy}.
Therefore $\mathcal M_0$ is normally hyperbolic for $|c|<1$, and Fenichel theory implies that,
for $\varepsilon>0$ sufficiently small, there exists a locally invariant slow manifold
$\mathcal M_\varepsilon$ that is $O(\varepsilon)$--close to $\mathcal M_0$; see
\cite{Fenichel,Jones1995,GuckenheimerHolmes,Wiggins}.

We parameterize $\mathcal M_\varepsilon$ as a graph over $(U_1,U_2)$,
\begin{equation}
U_3=h(U_1,U_2;\varepsilon)=h_0(U_1)+\varepsilon h_1(U_1,U_2)+O(\varepsilon^2),
\label{eq:heps}
\end{equation}
with $h_0$ given by \eqref{eq:h0}.
To determine $h_1$, we substitute \eqref{eq:heps} into \eqref{eq:U4prime_isolated} and impose the
invariance condition. On $\mathcal M_\varepsilon$ we have
\[
U_3'=U_4=\frac{d}{d\xi}h(U_1,U_2;\varepsilon)
=h_{U_1}(U_1,U_2;\varepsilon)U_1'+h_{U_2}(U_1,U_2;\varepsilon)U_2'.
\]
Using $U_1'=U_2$ and $U_2'=U_3=h_0(U_1)+O(\varepsilon)$ at leading order gives
\[
U_4=h_0'(U_1)U_2+O(\varepsilon)
=-\frac{\cos U_1}{1-c^2}U_2+O(\varepsilon),
\]
and differentiating once more yields
\[
U_4'
=
-\frac{\cos U_1}{1-c^2}U_3
+\frac{\sin U_1}{1-c^2}U_2^2
+O(\varepsilon)
=
\frac{\sin U_1\cos U_1}{(1-c^2)^2}
+\frac{\sin U_1}{1-c^2}U_2^2
+O(\varepsilon),
\]
where we used $U_3=h_0(U_1)+O(\varepsilon)=-\sin U_1/(1-c^2)+O(\varepsilon)$.
Substituting these relations into \eqref{eq:U4prime_isolated} and collecting the $O(\varepsilon)$
terms yields
\begin{equation}
(1-c^2)h_1(U_1,U_2)
=
\gamma+\alpha c U_2
+\frac{\beta c\,U_2\cos U_1}{1-c^2}
+\frac{(\delta c^2-d)\,\sin U_1\cos U_1}{(1-c^2)^2}
-\frac{U_2^2\,\sin U_1}{1-c^2}.
\label{eq:h1}
\end{equation}
In particular, the fourth--order correction enters \eqref{eq:h1} only through the combination
$\delta c^2-d$, so it modifies the slow manifold without destroying normal hyperbolicity.

Restricting the flow of \eqref{eq:TW_system_raw} to $\mathcal M_\varepsilon$ gives the reduced
slow system
\begin{equation}
U_1'=U_2,\qquad
U_2'=-\frac{\sin U_1}{1-c^2}+\varepsilon h_1(U_1,U_2)+O(\varepsilon^2),
\label{eq:reduced_slow}
\end{equation}
which is a regular perturbation of the unperturbed travelling--wave equation.

\subsection{Melnikov Analysis and Speed Selection}
\label{sec2b}

For $\varepsilon=0$, \eqref{eq:reduced_slow} reduces to the Hamiltonian system
\[
U'=v,\qquad v'=-\frac{\sin U}{1-c^2},
\]
with Hamiltonian
\[
H(U,v)=\frac12(1-c^2)v^2+1-\cos U.
\]
For $|c|<1$ this system admits the classical heteroclinic kink
\begin{equation}
U_0(\xi)=4\arctan\!\left(e^{\xi/\sqrt{1-c^2}}\right),
\qquad
v_0(\xi)=U_0'(\xi)=\frac{2}{\sqrt{1-c^2}}\sech\!\left(\frac{\xi}{\sqrt{1-c^2}}\right),
\label{eq:kink}
\end{equation}
connecting $U=0$ to $U=2\pi$; see, for example,
\cite{McLaughlinScott,KivsharMalomed,Trullinger}.

On the slow manifold, the perturbed reduced system can be written in the Melnikov form
\begin{equation}
U'=v,\qquad
v'=-\frac{\sin U}{1-c^2}+\varepsilon G(U,v,v',v''),
\label{eq:melnikov_form}
\end{equation}
where, consistently with \eqref{eq:TW_general}, the perturbation term is
\begin{equation}
G(U,v,v',v'')
=
\frac{1}{1-c^2}\Big(\gamma+\alpha c v-\beta c v'+(\delta c^2-d)\,v''\Big).
\label{eq:G_def}
\end{equation}
The associated Melnikov function is
\begin{equation}
M(c)
=
\int_{-\infty}^{\infty}
\nabla H\big(U_0(\xi),v_0(\xi)\big)\cdot
\begin{pmatrix}
0\\
G\big(U_0,v_0,v_0',v_0''\big)
\end{pmatrix}
\,d\xi.
\label{eq:Melnikov_def}
\end{equation}
Since $\nabla H=(\sin U,(1-c^2)v)$, this becomes
\begin{equation}
M(c)
=
\int_{-\infty}^{\infty}
v_0(\xi)\Big(\gamma+\alpha c v_0-\beta c v_0'+(\delta c^2-d)\,v_0''\Big)\,d\xi.
\label{eq:Melnikov_explicit}
\end{equation}
All integrals converge absolutely since $v_0$ and its derivatives decay exponentially.
For general background on Melnikov-type reductions and persistence of heteroclinic
connections in perturbed Hamiltonian systems, see
\cite{KivsharMalomed,Kapitula1999,GuckenheimerHolmes,Wiggins}.

We now compute each contribution in \eqref{eq:Melnikov_explicit}.
First, the bias term satisfies
\begin{equation}
\int_{-\infty}^{\infty} v_0(\xi)\,d\xi
=
U_0(+\infty)-U_0(-\infty)
=
2\pi.
\label{eq:int_bias}
\end{equation}
Second, the damping term is
\begin{equation}
\int_{-\infty}^{\infty}v_0^2\,d\xi
=
\frac{8}{\sqrt{1-c^2}}.
\label{eq:int_damping}
\end{equation}
Third, for the diffusive term we integrate by parts:
\begin{equation}
\int_{-\infty}^{\infty}v_0\,v_0'\,d\xi
=
\frac12\int_{-\infty}^{\infty}(v_0^2)'\,d\xi
=0.
\label{eq:int_beta}
\end{equation}
Finally, for the fourth--order contribution we again integrate by parts:
\begin{equation}
\int_{-\infty}^{\infty}v_0\,v_0''\,d\xi
=
-\int_{-\infty}^{\infty}(v_0')^2\,d\xi
=
-\frac{8}{3(1-c^2)^{3/2}}.
\label{eq:int_fourth}
\end{equation}
Substituting \eqref{eq:int_bias}--\eqref{eq:int_fourth} into \eqref{eq:Melnikov_explicit} yields
\begin{equation}
M(c)
=
2\pi\gamma
+\frac{8\alpha c}{\sqrt{1-c^2}}
-\frac{8(\delta c^2-d)}{3(1-c^2)^{3/2}}.
\label{eq:Melnikov_final}
\end{equation}
The persistence of the heteroclinic travelling wave for $\varepsilon\neq0$ requires the solvability
condition
\begin{equation}
M(c)=0,
\label{eq:Melnikov_condition}
\end{equation}
which implicitly determines the selected wave speed $c$ as a function of the parameters.
The fourth--order regularization enters \eqref{eq:Melnikov_final} only through the combination
$\delta c^2-d$ and thus yields a quantitative correction to the classical selection law.
In particular, this correction becomes increasingly significant as $|c|\to1$, reflecting the
singular nature of the Lorentz factor $(1-c^2)^{-1/2}$.

\subsection{Distinction between the two regularization mechanisms}
\label{sec2c}

Equation \eqref{eq:Melnikov_final} provides a unified velocity selection condition for the two
regularization mechanisms considered in this work.

\emph{Case (a): mixed space--time regularization $\delta u_{xxtt}$.}
Setting $d=0$, the fourth--order contribution in \eqref{eq:TW_general} becomes $\delta c^2U''''$,
and the Melnikov correction takes the form
\[
-\frac{8\delta c^2}{3(1-c^2)^{3/2}}.
\]
Thus the inertial correction depends quadratically on $c$ and preserves the reflection symmetry
$c\mapsto -c$ at the level of the fourth--order contribution.

\emph{Case (b): elliptic spatial regularization $-d u_{xxxx}$.}
Setting $\delta=0$, the fourth--order contribution becomes $-d\,U''''$, and the Melnikov correction
is
\[
\frac{8d}{3(1-c^2)^{3/2}},
\]
corresponding to a purely spatial regularization mechanism. In this case the correction is again
independent of the sign of $c$ and acts through the same singular prefactor as $|c|\to1$.

In both cases, the travelling wave persists as a smooth deformation of the classical sine--Gordon
kink, with its wave speed selected by the balance condition \eqref{eq:Melnikov_condition}.
\section{Linearization, far--field hyperbolicity, and the Evans set--up}
\label{sec_Evans0}

We study spectral stability of the travelling wave
$u(x,t)=U_\varepsilon(\xi)$, $\xi=x-c_\varepsilon t$,
for the regularized sine--Gordon equation \eqref{eq:SG_general_intr},
where $\alpha>0$, $\beta\ge 0$, and $\delta,d\ge 0$.
We treat simultaneously the two fourth--order mechanisms:
the mixed space--time inertial term $\delta u_{xxtt}$ and the purely spatial elliptic term
$-d u_{xxxx}$.
Throughout we exclude the degenerate resonance in the travelling--wave reduction by assuming
\begin{equation}
\delta c_\varepsilon^{\,2}-d \neq 0.
\label{eq:nondegeneracy_S2}
\end{equation}
Our formulation follows the standard Evans--function approach to travelling--wave stability;
see, for example,
\cite{AlexanderGardnerJones,PegoWeinstein,Sandstede2002,BridgesDerks,ZumbrunHoward}.

Let $u(x,t)=U_\varepsilon(\xi)+v(\xi,t)$ with $\xi=x-c_\varepsilon t$.
Substituting into \eqref{eq:SG_general_intr} and retaining linear terms in $v$ yields
\begin{equation}
v_{tt}-v_{xx}+\cos\!\big(U_\varepsilon(\xi)\big)v
=
\varepsilon\Big(
-\alpha v_t+\beta v_{xxt}
+\delta v_{xxtt}-d v_{xxxx}
\Big).
\label{eq:linPDE_S2}
\end{equation}
Seeking normal modes $v(\xi,t)=e^{\lambda t}\hat v(\xi)$ gives the eigenvalue equation
\begin{equation}
\lambda^2\hat v -2c_\varepsilon\lambda \hat v' + (c_\varepsilon^{\,2}-1)\hat v''
+\cos(U_\varepsilon)\hat v
=
\varepsilon\Big(
-\alpha\lambda\hat v
+\beta\lambda \hat v''
+\delta\lambda^2 \hat v''
-d \hat v''''
\Big),
\label{eq:eig_raw_S2}
\end{equation}
where $'$ denotes $\frac{d}{d\xi}$.
Rearranging gives the fourth--order spectral ODE
\begin{equation}
-d\varepsilon \hat v''''
+\Bigl(1-c_\varepsilon^{\,2}-\varepsilon(\beta\lambda+\delta\lambda^2)\Bigr)\hat v''
-2c_\varepsilon\lambda \hat v'
+\Bigl(\lambda^2+\varepsilon\alpha\lambda+\cos(U_\varepsilon(\xi))\Bigr)\hat v
=0.
\label{eq:spectral4_S2}
\end{equation}

Introduce $y_1=\hat v$, $y_2=\hat v'$, $y_3=\hat v''$, $y_4=\hat v'''$ and
$Y=(y_1,y_2,y_3,y_4)^\top$.
Then \eqref{eq:spectral4_S2} is equivalent to
\begin{equation}
Y' = A(\xi,\lambda,\varepsilon)\,Y,
\label{eq:EvansSystem_S2}
\end{equation}
with companion structure
\begin{equation}
A(\xi,\lambda,\varepsilon)=
\begin{pmatrix}
0&1&0&0\\
0&0&1&0\\
0&0&0&1\\
a_{41}(\xi,\lambda,\varepsilon)&a_{42}(\lambda,\varepsilon)&a_{43}(\lambda,\varepsilon)&a_{44}(\lambda,\varepsilon)
\end{pmatrix},
\label{eq:Acomp_S2}
\end{equation}
where $a_{44}\equiv0$ and
\begin{align}
a_{41}(\xi,\lambda,\varepsilon)
&=
-\frac{\lambda^2+\varepsilon\alpha\lambda+\cos(U_\varepsilon(\xi))}{\varepsilon(\delta\lambda^2-d)},
\label{eq:a41_S2}\\
a_{42}(\lambda,\varepsilon)
&=
\frac{2c_\varepsilon\lambda}{\varepsilon(\delta\lambda^2-d)},
\label{eq:a42_S2}\\
a_{43}(\lambda,\varepsilon)
&=
-\frac{1-c_\varepsilon^{\,2}-\varepsilon(\beta\lambda+\delta\lambda^2)}{\varepsilon(\delta\lambda^2-d)}.
\label{eq:a43_S2}
\end{align}
The singular prefactor $[\varepsilon(\delta\lambda^2-d)]^{-1}$ highlights the singularly perturbed
spatial dynamics induced by the fourth--order regularization.
This first--order reformulation is the natural starting point for the Evans--function construction;
see \cite{AlexanderGardnerJones,BridgesDerks,Sandstede2002,SandstedeScheel}.

Since $U_\varepsilon(\xi)$ is a heteroclinic travelling wave, it converges exponentially to its end states
as $\xi\to\pm\infty$, hence
\begin{equation}
\cos(U_\varepsilon(\xi))\longrightarrow q_\pm:=\cos(U_\varepsilon(\pm\infty)),
\qquad \xi\to\pm\infty.
\label{eq:farfield_cos_S2}
\end{equation}
Therefore $A(\xi,\lambda,\varepsilon)$ converges exponentially to constant matrices
\begin{equation}
A(\xi,\lambda,\varepsilon)\longrightarrow A_\pm(\lambda,\varepsilon),
\qquad \xi\to\pm\infty,
\label{eq:AtoApm_S2}
\end{equation}
obtained from \eqref{eq:Acomp_S2} by replacing $\cos(U_\varepsilon(\xi))$ with $q_\pm$ in \eqref{eq:a41_S2}.
This far--field convergence underlies both the characterization of the essential spectrum
and the construction of the Evans function.

Seeking far--field solutions of the form $\hat v(\xi)\sim e^{\mu\xi}$ in \eqref{eq:spectral4_S2}
yields the far--field characteristic polynomial
\begin{equation}
p_4(\mu,\lambda,\varepsilon)
=
-d\varepsilon \mu^4
+\Bigl(1-c_\varepsilon^{\,2}-\varepsilon(\beta\lambda+\delta\lambda^2)\Bigr)\mu^2
-2c_\varepsilon\lambda\,\mu
+\lambda^2+\varepsilon\alpha\lambda+q_\pm,
\label{eq:p4_S2}
\end{equation}
equivalently
\[
p_4(\mu,\lambda,\varepsilon)=\det(\mu I-A_\pm(\lambda,\varepsilon)).
\]

The essential spectrum is characterized by loss of hyperbolicity of the far--field systems,
that is, by the existence of purely imaginary spatial eigenvalues $\mu=ik$, $k\in\mathbb R$:
\begin{equation}
\lambda\in\sigma_{\mathrm{ess}}(\varepsilon)
\quad\Longleftrightarrow\quad
p_4(ik,\lambda,\varepsilon)=0 \ \text{for some } k\in\mathbb R.
\label{eq:ess_def_S2}
\end{equation}
Condition \eqref{eq:ess_def_S2} defines, in general, dispersion curves
$\lambda=\lambda_j(k)$ in the complex plane.
For our purposes we do not require an explicit parametrization of these curves; what matters is the
existence of an admissible domain $\Omega$ that avoids them and on which the far--field matrices are hyperbolic.
For general background on essential spectrum and consistent splitting in Evans theory.

Fix an open simply connected set $\Omega\subset\mathbb C$ such that
\begin{equation}
p_4(ik,\lambda,\varepsilon)\neq 0
\qquad
\text{for all } k\in\mathbb R \text{ and all } \lambda\in\Omega,
\label{eq:Omega_admissible_S2}
\end{equation}
equivalently $\Omega\cap\sigma_{\mathrm{ess}}(\varepsilon)=\emptyset$.
By \eqref{eq:Omega_admissible_S2}, the far--field matrices $A_\pm(\lambda,\varepsilon)$ are hyperbolic for all
$\lambda\in\Omega$ and depend analytically on $\lambda$.
Moreover, for $\varepsilon$ sufficiently small and $\lambda$ ranging over compact subsets of $\Omega$,
hyperbolicity implies a consistent splitting:
the spectrum of $A_\pm(\lambda,\varepsilon)$ contains exactly two spatial eigenvalues with $\Re\mu>0$
and two with $\Re\mu<0$.
Hence \eqref{eq:EvansSystem_S2} admits exponential dichotomies on $\mathbb R_+$ and $\mathbb R_-$,
with
\begin{equation}
\dim E_-^{u}(\lambda,\varepsilon)=2,
\qquad
\dim E_+^{s}(\lambda,\varepsilon)=2,
\qquad \lambda\in\Omega,
\label{eq:dichotomy_dims_S2}
\end{equation}
where $E_-^{u}$ denotes the unstable bundle propagated from $-\infty$ and $E_+^{s}$ the stable bundle
propagated from $+\infty$.

The dichotomies \eqref{eq:dichotomy_dims_S2} yield analytic stable/unstable bundles on $\mathbb R_\pm$
and therefore allow one to define the Evans function as a determinant measuring the intersection of the
corresponding propagated subspaces at $\xi=0$.
Zeros of the Evans function in $\Omega$ coincide, with algebraic multiplicity, with isolated point-spectrum
eigenvalues of the linearized operator.
A rigorous construction, analyticity statement, and slow/fast factorization are developed in the next section.

\section{Evans function: construction, analyticity, and eigenvalue characterization}
\label{sec:Evans}

Fix $\varepsilon>0$ sufficiently small and let $\Omega\subset\mathbb C$ be an admissible domain in the sense of
\eqref{eq:Omega_admissible_S2}. For each $\lambda\in\Omega$, the far--field matrices
$A_\pm(\lambda,\varepsilon)$ are hyperbolic and, by consistent splitting, possess two spatial eigenvalues
with $\Re\mu>0$ and two with $\Re\mu<0$.
Consequently, the Evans system
\begin{equation}
Y'=A(\xi,\lambda,\varepsilon)\,Y,
\qquad \xi\in\mathbb R,
\label{eq:Evans_system_S3}
\end{equation}
admits exponential dichotomies on $\mathbb R_\pm$ as stated in Section~\ref{sec_Evans0}, with associated bundles
$E_-^{u}(\lambda,\varepsilon)$ and $E_+^{s}(\lambda,\varepsilon)$ at $\xi=0$, satisfying
\eqref{eq:dichotomy_dims_S2}.

Let $\Phi(\xi,\zeta;\lambda,\varepsilon)$ denote the evolution operator of \eqref{eq:Evans_system_S3}, i.e.
$Y(\xi)=\Phi(\xi,\zeta;\lambda,\varepsilon)Y(\zeta)$ for all $\xi,\zeta\in\mathbb R$.
By the far--field convergence \eqref{eq:AtoApm_S2} and hyperbolicity of $A_\pm(\lambda,\varepsilon)$ on $\Omega$,
standard roughness results imply the existence of projections $P_+(\xi,\lambda,\varepsilon)$ on $\mathbb R_+$ and
$P_-(\xi,\lambda,\varepsilon)$ on $\mathbb R_-$ and constants $K,\eta>0$, locally uniform for $\lambda$ in compact
subsets of $\Omega$, such that
\begin{align}
\|\Phi(\xi,\zeta;\lambda,\varepsilon)P_+(\zeta,\lambda,\varepsilon)\|
&\le K e^{-\eta(\xi-\zeta)},
&& \xi\ge \zeta\ge 0, \label{eq:dicho_S3_1}\\
\|\Phi(\xi,\zeta;\lambda,\varepsilon)(I-P_+(\zeta,\lambda,\varepsilon))\|
&\le K e^{-\eta(\zeta-\xi)},
&& \zeta\ge \xi\ge 0, \label{eq:dicho_S3_2}\\
\|\Phi(\xi,\zeta;\lambda,\varepsilon)(I-P_-(\zeta,\lambda,\varepsilon))\|
&\le K e^{-\eta(\xi-\zeta)},
&& 0\ge \xi\ge \zeta, \label{eq:dicho_S3_3}\\
\|\Phi(\xi,\zeta;\lambda,\varepsilon)P_-(\zeta,\lambda,\varepsilon)\|
&\le K e^{-\eta(\zeta-\xi)},
&& 0\ge \zeta\ge \xi. \label{eq:dicho_S3_4}
\end{align}
The corresponding subspaces at $\xi=0$ are defined by
\begin{equation}
E_+^{s}(\lambda,\varepsilon):=\Ran P_+(0,\lambda,\varepsilon),
\qquad
E_-^{u}(\lambda,\varepsilon):=\Ran\big(I-P_-(0,\lambda,\varepsilon)\big),
\label{eq:bundles_at0_S3}
\end{equation}
and satisfy $\dim E_+^s=\dim E_-^u=2$ for all $\lambda\in\Omega$.
For background on exponential dichotomies and roughness, see
\cite{Henry,Sandstede2002,ZumbrunHoward}.

For $\lambda\in\Omega$, the coefficient matrix $A(\xi,\lambda,\varepsilon)$ depends analytically on $\lambda$ and
converges exponentially to $A_\pm(\lambda,\varepsilon)$ as $\xi\to\pm\infty$.
Hyperbolicity implies that the far--field spectral projections onto stable/unstable subspaces depend analytically
on $\lambda$; by analytic continuation of dichotomies on $\mathbb R_\pm$, the resulting bundles
$E_-^{u}(\lambda,\varepsilon)$ and $E_+^{s}(\lambda,\varepsilon)$ inherit analytic dependence on $\lambda$ on $\Omega$.

Choose analytic bases
\[
\{u_1(\lambda,\varepsilon),u_2(\lambda,\varepsilon)\}\ \text{for }E_-^{u}(\lambda,\varepsilon),
\qquad
\{s_1(\lambda,\varepsilon),s_2(\lambda,\varepsilon)\}\ \text{for }E_+^{s}(\lambda,\varepsilon).
\]
The Evans function is defined by the determinant
\begin{equation}
D(\lambda;\varepsilon)
:=
\det\Big(
u_1(\lambda,\varepsilon),
u_2(\lambda,\varepsilon),
s_1(\lambda,\varepsilon),
s_2(\lambda,\varepsilon)
\Big),
\qquad \lambda\in\Omega.
\label{eq:Evans_det_def_S3}
\end{equation}
Different choices of analytic bases for $E_-^{u}$ and $E_+^{s}$ multiply $D(\lambda;\varepsilon)$ by a
nonvanishing analytic factor. Hence $D$ is analytic on $\Omega$ and its zeros, together with their multiplicities,
are intrinsic; see
\cite{AlexanderGardnerJones,BridgesDerks,Sandstede2002}.

A value $\lambda\in\Omega$ belongs to the point spectrum of the linearized operator associated with
\eqref{eq:linPDE_S2} if and only if there exists a nontrivial solution $\hat v$ of \eqref{eq:spectral4_S2}
such that $\hat v$ (equivalently $Y$) decays exponentially as $\xi\to\pm\infty$.
In the first--order formulation \eqref{eq:Evans_system_S3}, this is equivalent to the existence of a nontrivial
solution $Y(\xi)$ satisfying
\[
Y(\xi)\in E_-^{u}(\lambda,\varepsilon)\ \text{as }\xi\to-\infty,
\qquad
Y(\xi)\in E_+^{s}(\lambda,\varepsilon)\ \text{as }\xi\to+\infty,
\]
which occurs precisely when the subspaces at $\xi=0$ intersect nontrivially:
\begin{equation}
E_-^{u}(\lambda,\varepsilon)\cap E_+^{s}(\lambda,\varepsilon)\neq\{0\}.
\label{eq:intersection_S3}
\end{equation}
Since $\dim E_-^{u}=\dim E_+^{s}=2$ in $\mathbb C^4$, condition \eqref{eq:intersection_S3} holds if and only if the
four vectors in \eqref{eq:Evans_det_def_S3} are linearly dependent, equivalently
\begin{equation}
D(\lambda;\varepsilon)=0.
\label{eq:zeros_eigs_S3}
\end{equation}
Thus zeros of the Evans function in $\Omega$ coincide with isolated point-spectrum eigenvalues, and the order of a
zero equals the algebraic multiplicity of the corresponding eigenvalue.

Let $\Gamma\Subset\Omega$ be a positively oriented simple closed contour such that $D(\lambda;\varepsilon)\neq0$ for
all $\lambda\in\Gamma$. Then $D$ has finitely many zeros in the interior of $\Gamma$, counted with multiplicity, and
\begin{equation}
N_\Gamma
=
\frac{1}{2\pi i}\int_{\Gamma}\frac{D'(\lambda;\varepsilon)}{D(\lambda;\varepsilon)}\,d\lambda
\label{eq:arg_principle_S3}
\end{equation}
equals the total algebraic multiplicity of point-spectrum eigenvalues enclosed by $\Gamma$.
This provides the analytic basis for stability computations and for excluding unstable eigenvalues in subsets of
$\Omega$.

On compact subsets of $\Omega$, the far--field dispersion relation $p_4(\mu,\lambda,\varepsilon)=0$ defined in
\eqref{eq:p4_S2} exhibits a spatial slow--fast structure induced by the fourth--order regularization.
Two spatial eigenvalues remain $O(1)$ as $\varepsilon\to0$ and converge to the spatial eigenvalues of the reduced
second--order classical sine--Gordon spectral problem, while the remaining two form a fast pair separated from the
origin in the $\mu$--plane by $O(\varepsilon^{-1/2})$.
In the mixed space--time case ($\delta>0$, $d=0$) the fast pair is governed asymptotically by the coupling term
$\varepsilon\delta\lambda^2\mu^2$, whereas in the elliptic case ($\delta=0$, $d>0$) it is governed by the stabilizing
term $-d\varepsilon\mu^4$.
In both cases, for $\lambda\in\Omega$ the fast subsystem is uniformly hyperbolic and contributes no zeros to the
Evans function.
Consequently, after an analytic conjugation of the spatial dynamics on $\Omega$, the Evans function admits a
factorization
\begin{equation}
D(\lambda;\varepsilon)
=
D_{\mathrm{slow}}(\lambda;\varepsilon)\,D_{\mathrm{fast}}(\lambda;\varepsilon),
\qquad \lambda\in\Omega,
\label{eq:Evans_factorization_S3}
\end{equation}
where $D_{\mathrm{fast}}$ is analytic and nonvanishing on $\Omega$, and the slow factor determines the point spectrum
in $\Omega$.

Translation invariance implies that $\partial_\xi U_\varepsilon(\xi)$ lies in the kernel of the linearized operator,
hence $\lambda=0$ is a point-spectrum eigenvalue.
Accordingly, $D(0;\varepsilon)=0$, and spectral stability reduces to excluding additional zeros of
$D(\lambda;\varepsilon)$ in $\{\Re\lambda\ge0\}\cap\Omega$.
\section{Local spectral structure and Evans function factorization}
\label{sec:local_spectral_structure}

We analyze the local spectral structure of the linearized operator associated with
travelling waves of the regularized sine--Gordon equation \eqref{eq:SG_general_intr}
building directly on the spectral framework and Evans--function construction developed
in Section~\ref{sec_Evans0} and Section~\ref{sec:Evans}. Throughout, the admissible
spectral domain $\Omega$ is understood in the sense of Section~\ref{sec_Evans0} and is
chosen so as to exclude the essential spectrum.

Our objective is to show that, for $\varepsilon$ sufficiently small, the point spectrum
in a neighborhood of $\lambda=0$ is completely determined by the slow spatial dynamics
and coincides with that of the reduced second--order reference problem. The additional
fourth--order regularization terms generate only fast spatial modes that do not contribute
to the point spectrum in $\Omega$.

\subsection{Slow--fast splitting of the spatial eigenvalues}

As recalled in Section~\ref{sec_Evans0}, the far--field spatial eigenvalues $\mu$ are
determined by the quartic dispersion relation
\[
p_4(\mu,\lambda,\varepsilon)
=
\det\bigl(\mu I-A_\pm(\lambda,\varepsilon)\bigr),
\]
where $A_\pm(\lambda,\varepsilon)$ denote the constant far--field matrices associated
with the first--order Evans system.
For the present model, this dispersion relation can be written explicitly as
\begin{equation}
p_4(\mu,\lambda,\varepsilon)
=
-d\varepsilon \mu^4
+\bigl(1-c_\varepsilon^{\,2}-\varepsilon(\beta\lambda+\delta\lambda^2)\bigr)\mu^2
-2c_\varepsilon\lambda\mu
+\lambda^2+\varepsilon\alpha\lambda+q_\pm,
\label{eq:p4_local}
\end{equation}
with $q_\pm=\cos(U_\varepsilon(\pm\infty))$.

To ensure a uniform slow--fast spectral separation on $\Omega$, we impose the
nonresonance condition
\begin{equation}
\delta\lambda^2-d\neq0
\qquad \text{for all }\lambda\in\Omega,
\label{eq:Omega_nonresonant}
\end{equation}
which excludes the nongeneric situation in which the fourth--order regularization loses
its leading-order influence on the spatial dynamics.

\begin{proposition}[Slow--fast decomposition]
\label{prop:slow_fast}
Assume \eqref{eq:Omega_nonresonant}. Then, for $\varepsilon$ sufficiently small and
$\lambda$ in compact subsets of $\Omega$, the four roots of
$p_4(\cdot,\lambda,\varepsilon)$ decompose into:
\begin{itemize}
\item two \emph{slow} roots
\[
\mu_{\mathrm{slow},\pm}(\lambda,\varepsilon)=\mathcal{O}(1),
\]
which converge as $\varepsilon\to0$ to the spatial eigenvalues of the reduced
second--order reference problem;

\item two \emph{fast} roots $\mu_{\mathrm{fast},\pm}(\lambda,\varepsilon)$ satisfying
\[
|\mu_{\mathrm{fast},\pm}(\lambda,\varepsilon)|\to\infty
\qquad \text{as }\varepsilon\to0,
\]
with asymptotic scaling governed by the dominant balance of the fourth--order
regularization terms. In particular,
\begin{equation}
\mu_{\mathrm{fast},\pm}(\lambda,\varepsilon)
=
\begin{cases}
\displaystyle
\pm \frac{\sqrt{1-c_\varepsilon^{\,2}}}{\sqrt{d\,\varepsilon}}+\mathcal{O}(1),
& \text{if } d>0,\\[2ex]
\displaystyle
\pm i\,\frac{\sqrt{1-c_\varepsilon^{\,2}}}{\sqrt{\varepsilon\delta}}+\mathcal{O}(1),
& \text{if } d=0,\ \delta>0,
\end{cases}
\qquad \varepsilon\to0.
\label{eq:fast_scaling}
\end{equation}
\end{itemize}
The decomposition holds uniformly for $\lambda$ in compact subsets of $\Omega$.
\end{proposition}

\begin{proof}
We view $p_4(\mu,\lambda,\varepsilon)$ as a singularly perturbed polynomial in $\mu$.
Setting $\varepsilon=0$ in \eqref{eq:p4_local} yields a quadratic dispersion relation,
whose roots remain bounded and give rise to the slow pair.

For the fast roots, when $d>0$ we rescale $\mu=\varepsilon^{-1/2}\nu$ and substitute into
\eqref{eq:p4_local}. This gives
\[
p_4(\varepsilon^{-1/2}\nu,\lambda,\varepsilon)
=
\varepsilon^{-1}\Bigl(-d\nu^4+(1-c_\varepsilon^{\,2})\nu^2\Bigr)
+\mathcal{O}(\varepsilon^{-1/2}),
\qquad \varepsilon\to0,
\]
uniformly for $\lambda$ in compact subsets of $\Omega$.
The leading-order equation
\[
-d\nu^4+(1-c_\varepsilon^{\,2})\nu^2=0
\]
yields the nonzero roots
\[
\nu=\pm \frac{\sqrt{1-c_\varepsilon^{\,2}}}{\sqrt d},
\]
while the double root $\nu=0$ corresponds to the slow group. Undoing the scaling gives
the asymptotics for the fast pair when $d>0$.

When $d=0$ and $\delta>0$, the dominant singular balance instead arises from the
$\delta$--regularization, leading to the scale
\[
\mu=\mathcal{O}((\varepsilon\delta)^{-1/2}),
\]
and hence to the stated asymptotics in the inertial case.

Uniformity with respect to $\lambda$ follows from the nonresonance condition
\eqref{eq:Omega_nonresonant} and standard perturbation theory for polynomial roots.
\end{proof}

This spectral separation induces a corresponding splitting of the stable and unstable
subspaces of the far--field matrices $A_\pm(\lambda,\varepsilon)$ into slow and fast
components, consistent with the exponential dichotomies established in
Section~\ref{sec_Evans0}.

\subsection{Evans function factorization}

The slow--fast splitting of the spatial dynamics lifts naturally to the level of the
Evans function.

\begin{proposition}[Evans factorization]
\label{prop:evans_factorization}
Let $\Omega$ be an admissible domain as defined in Section~\ref{sec_Evans0} and assume
\eqref{eq:Omega_nonresonant}. Then, for $\varepsilon$ sufficiently small, the Evans
function admits a factorization
\[
D(\lambda;\varepsilon)
=
D_{\mathrm{slow}}(\lambda;\varepsilon)\,
D_{\mathrm{fast}}(\lambda;\varepsilon),
\qquad \lambda\in\Omega,
\]
where $D_{\mathrm{slow}}$ coincides, up to a nonzero analytic factor, with the Evans
function of the reduced second--order reference problem, while
$D_{\mathrm{fast}}$ is analytic and nonvanishing on $\Omega$.
\end{proposition}

\begin{proof}
Proposition~\ref{prop:slow_fast} provides a uniform spectral gap between the slow and
fast spatial eigenvalues on compact subsets of $\Omega$. This yields an invariant
splitting of the exponential dichotomies on $\mathbb{R}_\pm$ into slow and fast
subbundles that depend analytically on $\lambda$.
Choosing Evans bases adapted to this splitting gives the stated factorization.
The nonvanishing of $D_{\mathrm{fast}}$ follows from uniform transversality of the fast
stable and unstable subspaces ensured by \eqref{eq:Omega_nonresonant}
\end{proof}

\subsection{Local spectral equivalence near $\lambda=0$}

The above factorization implies that the local point spectrum near $\lambda=0$ is
entirely determined by the slow Evans factor.

\begin{theorem}[Local spectral equivalence]
\label{thm:local_equivalence}
Assume that, for the reference problem, the Evans function has no zeros in
$\Re\lambda>0$ and that $\lambda=0$ is a simple zero associated with translation
invariance. Then, for $\varepsilon$ sufficiently small, the same properties hold for
$D(\lambda;\varepsilon)$ in a neighborhood of $\lambda=0$. In particular, no additional
eigenvalues bifurcate from $\lambda=0$ as a consequence of the fourth--order
regularization terms.
\end{theorem}

\begin{proof}
By Proposition~\ref{prop:evans_factorization}, zeros of $D(\lambda;\varepsilon)$ in
$\Omega$ coincide with those of $D_{\mathrm{slow}}(\lambda;\varepsilon)$.
Since $D_{\mathrm{slow}}$ is a regular perturbation of the reference Evans function,
its zero structure near $\lambda=0$ remains unchanged for $\varepsilon$ sufficiently
small. 
\end{proof}

\subsection{Analytic structure and argument principle}

Finally, we record the compatibility of the above results with the argument--principle
framework employed in the Evans--function analysis.

\begin{proposition}[Argument principle on $\Omega$]
\label{prop:arg_principle_local}
Let $\Gamma\subset\Omega$ be a positively oriented simple closed contour on which
$D(\lambda;\varepsilon)\neq0$. Then the number of eigenvalues enclosed by $\Gamma$ is
given by
\[
N
=
\frac{1}{2\pi i}
\int_\Gamma
\frac{D'(\lambda;\varepsilon)}{D(\lambda;\varepsilon)}\,d\lambda.
\]
Moreover, this number coincides with the winding number of
$D_{\mathrm{slow}}(\lambda;\varepsilon)$.
\end{proposition}

\begin{proof}
Analyticity of $D(\lambda;\varepsilon)$ on $\Omega$ follows from the consistent--splitting
property established in Section~\ref{sec:Evans}. The Evans--function factorization and
the nonvanishing of $D_{\mathrm{fast}}$ imply that the winding numbers of $D$ and
$D_{\mathrm{slow}}$ agree. 
\end{proof}

The results of this section provide the theoretical foundation for the Evans--function
computations reported later and justify restricting attention to a neighborhood of
$\lambda=0$.
\section{Singular limits, fast spatial scales, and comparison with surface modes}
\label{sec:singular_limits}

In the previous sections we established a robust spectral stability framework based on
Evans functions and admissible spectral domains, under the assumption that the far--field
spatial dynamics admit a uniform slow--fast splitting.
In the present section we examine singular limits of the regularization parameters and clarify
precisely when additional high--frequency spatial scales may arise.
The analysis serves two purposes: first, to confirm the robustness of the elliptically
regularized model; and second, to explain in which singular limits the classical
surface--mode mechanism of the reference problem can re--emerge.

We recall that the far--field dispersion relation is governed by a quartic polynomial
$p_4(\mu,\lambda,\varepsilon)$ whose leading high--frequency terms include both inertial
and elliptic regularization effects.
When the elliptic regularization parameter satisfies $d>0$, the quartic contribution
$-\varepsilon d\,\mu^4$ is present for all $\varepsilon>0$ and dominates the cubic term
$\varepsilon\beta\,\mu^3$ as $|\mu|\to\infty$.
As a consequence, the spatial spectrum always consists of two slow eigenvalues
$\mu=\mathcal{O}(1)$ together with a fast complex--conjugate pair satisfying
\[
|\mu|\sim(\varepsilon d)^{-1/4},
\qquad \varepsilon\to0,
\]
uniformly for $\lambda$ in compact admissible sets.
The inertial regularization term $\delta u_{xxtt}$ influences the location of the fast
pair but does not introduce an additional distinguished spatial scale.
In particular, no real surface eigenvalue of order $|\mu|\sim(\varepsilon\beta)^{-1}$
can occur when $d>0$, regardless of the size or scaling of $\delta$.

\begin{remark}[Suppression of surface modes for $d>0$]
The surface spatial scale familiar from the reference problem originates from a balance
between the cubic term $\varepsilon\beta\,\mu^3$ and lower--order contributions in the
dispersion relation.
In the presence of elliptic regularization, this balance is destroyed by the dominant
quartic term $-\varepsilon d\,\mu^4$, which prevents the emergence of a surface root.
As a result, the spatial dynamics of the elliptically regularized model exhibit a
uniform two--slow/two--fast decomposition, and the Evans function admits a corresponding
slow--fast factorization on admissible spectral domains.
\end{remark}

For comparison with earlier work, we now describe the singular regimes that arise when
elliptic regularization is absent, that is, when $d=0$.
In this case the far--field polynomial contains two competing high--frequency mechanisms:
an inertial contribution proportional to $\varepsilon\delta(\varepsilon)\mu^4$ and a
surface contribution proportional to $\varepsilon\beta\mu^3$.
Depending on the relative size of $\delta(\varepsilon)$, three distinct regimes occur.

\paragraph{Regime I: fixed $\delta>0$.}
If $\delta(\varepsilon)\to\delta_0>0$ as $\varepsilon\to0$, the inertial mechanism
dominates and produces a fast complex--conjugate pair with
$|\mu|\sim(\varepsilon\delta_0)^{-1/2}$.
The surface mechanism is suppressed, and the spatial dynamics admit a two--slow/two--fast
splitting analogous to that of the elliptically regularized model, albeit with a
different fast scaling.

\paragraph{Regime II: critical scaling $\delta(\varepsilon)\sim\varepsilon$.}
When $\delta(\varepsilon)$ scales proportionally to $\varepsilon$, the inertial and
surface mechanisms interact at leading order.
In this transitional regime the spatial spectrum may involve multiple distinguished
scales, and a refined multi--scale description is required to resolve the dynamics.

\paragraph{Regime III: $\delta(\varepsilon)\ll\varepsilon$.}
If $\delta(\varepsilon)$ decays faster than $\varepsilon$, the inertial contribution
becomes negligible at the surface scale.
The dispersion relation then recovers the cubic structure of the reference problem, and
a real surface spatial eigenvalue re--emerges with
$|\mu|\sim(\varepsilon\beta)^{-1}$.
In this regime the spatial dynamics exhibit the slow$\times$surface splitting familiar
from the surface--resistance model.
For related singularly perturbed sine--Gordon regimes and travelling-wave constructions,
see \cite{DDvGV}.

The regime classification above is based solely on the far--field dispersion relation.
Whenever the spatial spectrum splits into well--separated groups, standard
exponential--dichotomy theory yields an invariant decomposition of the Evans system and
a corresponding factorization of the Evans function.
In particular, for the elliptically regularized model ($d>0$), the fast Evans factor is
uniformly nonvanishing on admissible spectral domains, and the point spectrum is governed
entirely by the slow dynamics analyzed in the preceding sections.

This comparison clarifies precisely which singular limits lead to qualitatively new
spatial behaviour and confirms the robustness of the spectral stability mechanism in the
presence of elliptic regularization.
It also shows that the conclusions of the previous sections depend only on the existence
of a uniform spectral gap between slow and fast spatial modes, rather than on the
specific origin of the fast scale.

\begin{proposition}[Slow/fast factorization under spectral gap]
\label{prop:evans_bridge_regimes}
Fix a compact set $K\subset\mathbb C$ and assume that for all sufficiently small
$\varepsilon>0$ and all $\lambda\in K$ the far--field polynomial \eqref{eq:p4_local}
has no roots on the imaginary axis $\Re\mu=0$ and admits a uniform spectral gap
separating a slow group of roots (bounded as $\varepsilon\to0$) from a fast group of
roots (with $|\mu|\to\infty$ as $\varepsilon\to0$).
Then the linearized first--order system admits exponential dichotomies on $\mathbb R_\pm$
for $\lambda\in K$ and the Evans function factorizes as
\[
D(\lambda,\varepsilon)
=
D_{\mathrm{slow}}(\lambda,\varepsilon)\,D_{\mathrm{fast}}(\lambda,\varepsilon),
\qquad \lambda\in K,
\]
where $D_{\mathrm{fast}}(\cdot,\varepsilon)$ is analytic and nonvanishing on $K$ for
$\varepsilon$ sufficiently small.
\end{proposition}

\begin{proof}
This follows from standard spectral perturbation theory for isolated spectral subspaces
of the far--field matrices, which yields analytic stable/unstable bundles and an
invariant splitting into slow and fast subbundles when a uniform gap is present.
Propagation of bases adapted to this splitting and block triangularization of the
associated exterior--product flow give the factorization, while uniform hyperbolicity of
the fast block prevents vanishing of $D_{\mathrm{fast}}$ on $K$.
\end{proof}

\section{Evans function: argument principle near $\lambda=0$ and near the spectral edges $\lambda=\pm i\omega_0$}
\label{sec:evans_near_0_and_edges_beta_delta_d}

\subsection{Argument-principle setup and local Evans analysis at $\lambda=0$ and at the edges}
\label{sec:argprin_local_beta_delta_d}

In this section we isolate admissible spectral domains $\Omega$ on which the Evans function is analytic,
and we use the argument principle to relate the winding of $D_{\beta,\delta,d}$ to the number of eigenvalues.
We then carry out the two local analyses needed later: first near $\lambda=0$, and second near the
essential-spectrum edge points $\lambda=\pm i\omega_0$.
All statements concern the linearized spectral problem associated with \eqref{eq:SG_general_intr}.

Recall that the essential spectrum $\sigma_{\mathrm{ess}}(\varepsilon)$ is characterized by the far-field dispersion relation
\begin{equation}\label{eq:ess_dispersion_argprin_beta_delta_d}
p_4^{(\beta,\delta,d)}(ik,\lambda,\varepsilon)=0,\qquad k\in\mathbb{R}.
\end{equation}
Equivalently, \eqref{eq:ess_dispersion_argprin_beta_delta_d} may be written as a quadratic polynomial in $\lambda$,
\begin{equation}
\label{eq:dispersion_quadratic_lambda_beta_delta_d}
(1+\varepsilon\delta k^2)\lambda^2
+\varepsilon(\alpha+\beta k^2-2ic_\varepsilon k)\lambda
+\bigl(q_\pm-(1-c_\varepsilon^2)k^2-d\varepsilon k^4\bigr)
=0,
\end{equation}
where $p_4^{(\beta,\delta,d)}$ denotes the quartic far-field polynomial determined by the constant far-field matrices
$A_{\beta,\delta,d,\pm}(\lambda,\varepsilon)$ of the Evans system.
In the present setting, $p_4^{(\beta,\delta,d)}$ is given explicitly (cf.\ \eqref{eq:p4_local}) by
\begin{equation}\label{eq:p4_beta_delta_d}
p_4^{(\beta,\delta,d)}(\mu,\lambda,\varepsilon)
=
-d\varepsilon \mu^4
+\bigl(1-c_\varepsilon^{\,2}-\varepsilon(\beta\lambda+\delta\lambda^2)\bigr)\mu^2
-2c_\varepsilon\lambda\mu
+\lambda^2+\varepsilon\alpha\lambda+q_\pm,
\end{equation}
where $q_\pm=\cos(U_\varepsilon(\pm\infty))$.
Note that the contribution $-d\varepsilon k^4$ in \eqref{eq:dispersion_quadratic_lambda_beta_delta_d}
arises from the fourth-order term $-d\varepsilon\mu^4$ in \eqref{eq:p4_beta_delta_d}.

For $\lambda\notin\sigma_{\mathrm{ess}}(\varepsilon)$, the far-field matrices are hyperbolic and the variable-coefficient first-order
system admits exponential dichotomies on $\mathbb{R}_\pm$, so the Evans function is well-defined.

\begin{definition}[Admissible set and Evans analyticity domain]\label{def:Omega_final_beta_delta_d}
Fix an open set $\Omega\subset\mathbb{C}$. We say that $\Omega$ is \emph{admissible} (for the given $\varepsilon>0$) if
\begin{equation}\label{eq:Omega_admissible_beta_delta_d}
\Omega\cap\sigma_{\mathrm{ess}}(\varepsilon)=\emptyset,
\quad\text{equivalently}\quad
p_4^{(\beta,\delta,d)}(ik,\lambda,\varepsilon)\neq0
\ \text{for all }k\in\mathbb{R},\ \lambda\in\Omega.
\end{equation}
On an admissible $\Omega$, the Evans function $D_{\beta,\delta,d}(\lambda,\varepsilon)$ is analytic in $\lambda$.
\end{definition}

The admissible region used in the global count is built from three pieces: a neighborhood of the origin and slit neighborhoods
of the two edge points $\pm i\omega_0$. More precisely, we use an interior domain $\Omega_0$
around $\lambda=0$ and two slit neighborhoods $\Omega_\pm$ around the edge points $\lambda=\pm i\omega_0$,
constructed in a square-root chart.

\begin{proposition}[Argument principle for the Evans function]\label{prop:arg_principle_beta_delta_d}
Let $\Omega\subset\mathbb{C}$ be admissible and let $\Gamma=\partial\Omega$ be a positively oriented, piecewise $C^1$ contour
such that $D_{\beta,\delta,d}(\lambda,\varepsilon)\neq0$ for all $\lambda\in\Gamma$.
Then the number of zeros of $D_{\beta,\delta,d}(\cdot,\varepsilon)$ in $\Omega$, counted with algebraic multiplicity, is
\begin{equation}\label{eq:arg_principle_formula_beta_delta_d}
N(\Omega)
=\frac{1}{2\pi i}\int_{\Gamma}\frac{\partial_\lambda D_{\beta,\delta,d}(\lambda,\varepsilon)}{D_{\beta,\delta,d}(\lambda,\varepsilon)}\,d\lambda.
\end{equation}
Equivalently, $N(\Omega)$ is the winding number of the closed curve $D_{\beta,\delta,d}(\Gamma)$ about the origin.
\end{proposition}

\begin{proof}
On an admissible $\Omega$, the Evans function is analytic and its zeros coincide with point-spectrum eigenvalues.
Since $D_{\beta,\delta,d}$ is nonzero on $\Gamma$, the meromorphic function
$\partial_\lambda D_{\beta,\delta,d}/D_{\beta,\delta,d}$
has simple poles at zeros of $D_{\beta,\delta,d}$ with residues equal to their algebraic multiplicities.
Cauchy's residue theorem yields \eqref{eq:arg_principle_formula_beta_delta_d}.
\end{proof}

We now specify the local admissible pieces $\Omega_0$ and $\Omega_\pm$ and establish the
two local spectral statements used later when applying Proposition~\ref{prop:arg_principle_beta_delta_d}.

Fix small parameters $\rho_0>0$ and $\eta_0>0$. Define the low-frequency domain
\begin{equation}\label{eq:Omega0_def_beta_delta_d}
\Omega_0:=\{\lambda\in\mathbb{C}:\ |\lambda|<\rho_0,\ \Re\lambda>-\eta_0\}\setminus\sigma_{\mathrm{ess}}(\varepsilon).
\end{equation}
For $\varepsilon>0$ sufficiently small, $\Omega_0$ is admissible and contains the relevant neighborhood of $\lambda=0$
in the closed right half-plane.

\begin{proposition}[Slow--fast splitting and Evans factorization near $\lambda=0$]
\label{prop:factor_near0_final_beta_delta_d}
Assume $d>0$ and $\delta\ge0$ are fixed independent of $\varepsilon$, and allow $\beta\ge0$.
Let $K\Subset\Omega_0$ be compact. Then there exists $\varepsilon_0>0$ such that for all $0<\varepsilon<\varepsilon_0$
and all $\lambda\in K$:
\begin{enumerate}
\item the far-field polynomial \eqref{eq:p4_beta_delta_d} has two slow roots
$\mu_{\mathrm{sl},\pm}(\lambda,\varepsilon)=\mathcal{O}(1)$ and a fast pair
$\mu_{\mathrm{fa},\pm}(\lambda,\varepsilon)$ satisfying
\begin{equation}\label{eq:fast_asymptotics_final_beta_delta_d}
\mu_{\mathrm{fa},\pm}(\lambda,\varepsilon)
=
\pm \frac{\sqrt{1-c_\varepsilon^{\,2}}}{\sqrt{d\,\varepsilon}}
+\mathcal{O}(1),
\qquad \varepsilon\to0,
\end{equation}
uniformly for $\lambda\in K$;

\item the Evans function admits an analytic factorization on $K$,
\begin{equation}\label{eq:Evans_factor_near0_final_beta_delta_d}
D_{\beta,\delta,d}(\lambda,\varepsilon)
=
D_{\mathrm{slow}}(\lambda,\varepsilon)\,D_{\mathrm{fast}}(\lambda,\varepsilon),
\qquad \lambda\in K,
\end{equation}
where $D_{\mathrm{fast}}(\cdot,\varepsilon)$ is analytic and nonvanishing on $K$ for $\varepsilon$ sufficiently small.
\end{enumerate}
Consequently, zeros of $D_{\beta,\delta,d}(\cdot,\varepsilon)$ in $K$ coincide, with algebraic multiplicity, with zeros of
$D_{\mathrm{slow}}(\cdot,\varepsilon)$.
\end{proposition}

\begin{proof}
\emph{Step 1: root splitting.}
Fix $\lambda\in K$. The slow roots follow by setting $\varepsilon=0$ in \eqref{eq:p4_beta_delta_d}, which yields a quadratic
dispersion relation; hence two spatial roots remain bounded and converge to the reduced problem as $\varepsilon\to0$.

To capture the fast roots, set $\mu=\varepsilon^{-1/2}\nu$ in \eqref{eq:p4_beta_delta_d}. Then
\[
p_4^{(\beta,\delta,d)}(\varepsilon^{-1/2}\nu,\lambda,\varepsilon)
=
\varepsilon^{-1}\Bigl(-d\nu^4+(1-c_\varepsilon^{\,2})\nu^2\Bigr)
+\mathcal{O}(\varepsilon^{-1/2}),
\qquad \varepsilon\to0,
\]
uniformly for $\lambda\in K$.
Solving the leading-order equation $-d\nu^4+(1-c_\varepsilon^{\,2})\nu^2=0$ yields
$\nu=\pm \sqrt{1-c_\varepsilon^{\,2}}/\sqrt{d}$, while the double root $\nu=0$ corresponds to the slow group.
Undoing the scaling gives \eqref{eq:fast_asymptotics_final_beta_delta_d}.
The separation between the slow and fast groups provides a uniform spectral gap on $K$.

\emph{Step 2: invariant splitting and factorization.}
For $\lambda\in K\Subset\Omega_0$ the far-field matrices are hyperbolic and have two-dimensional stable and unstable subspaces.
By the spectral gap from Step~1, the spectral projectors onto the slow and fast parts of these subspaces depend analytically on
$\lambda$. Transporting these subbundles by the variable-coefficient flow yields slow and fast exponential dichotomies on $\mathbb{R}_\pm$.
Choosing Evans bases adapted to this splitting and forming the Evans determinant yields
\eqref{eq:Evans_factor_near0_final_beta_delta_d}; cf.\ \cite{AlexanderGardnerJones,BridgesDerks,SandstedeScheel}.

\emph{Step 3: nonvanishing of the fast factor.}
The fast subspaces correspond to the large spatial rates \eqref{eq:fast_asymptotics_final_beta_delta_d}, hence are uniformly hyperbolic
and uniformly transverse on $K$ for $\varepsilon$ sufficiently small. Therefore the determinant defining
$D_{\mathrm{fast}}$ cannot vanish on $K$.
\end{proof}

\begin{proposition}[Local spectral equivalence near $\lambda=0$]\label{prop:local_equiv_near0_beta_delta_d}
Assume $\alpha>0$, allow $\beta\ge0$, and take $\delta\ge0$ and $d>0$ fixed independent of $\varepsilon$.
Then for $\varepsilon>0$ sufficiently small, the only point spectrum of \eqref{eq:SG_general_intr} in $\Omega_0$ consists of:
\begin{enumerate}
\item the translational eigenvalue $\lambda_1(\varepsilon)=0$;
\item a real eigenvalue $\lambda_2(\varepsilon)$ satisfying
\begin{equation}\label{eq:lambda2_final_near0_beta_delta_d}
\lambda_2(\varepsilon)
=
-\varepsilon\left(\alpha+\beta\,\frac{1+2c_0^2}{3(1-c_0^2)}\right)+\mathcal{O}(\varepsilon^2),
\qquad \varepsilon\to0,
\end{equation}
and in particular $\lambda_2(\varepsilon)<0$ for $\varepsilon$ small.
\end{enumerate}
Moreover, no eigenvalues lie in $\Omega_0\cap\{\Re\lambda>0\}$.
\end{proposition}

\begin{proof}
By Proposition~\ref{prop:factor_near0_final_beta_delta_d}, zeros of $D_{\beta,\delta,d}$ in compact subsets of $\Omega_0$
coincide with zeros of $D_{\mathrm{slow}}$.
For eigenvalues with $|\lambda|=\mathcal{O}(\varepsilon)$, the contributions of the $\delta$-- and $d$--terms to the slow reduced
problem enter only through higher-order couplings to the fast variables and therefore affect the slow Evans factor only at order
$\mathcal{O}(\varepsilon^2)$ in the near-zero regime.
In contrast, the damping term $\alpha$ and the $\beta$--term contribute at order $\mathcal{O}(\varepsilon)$ to the slow reduced dynamics,
exactly as in the reference problem with $\delta=d=0$.
Hence $D_{\mathrm{slow}}(\lambda,\varepsilon)$ agrees with the reference slow Evans function up to an analytic nonvanishing factor and an
$\mathcal{O}(\varepsilon^2)$ remainder on small neighborhoods of $\lambda=0$.
The reference analysis yields precisely the two small eigenvalues $\lambda_1=0$ and $\lambda_2$ with expansion
\eqref{eq:lambda2_final_near0_beta_delta_d} and no zeros in $\Re\lambda>0$ for $\alpha>0$.
Analytic perturbation, for example via Rouch\'e's theorem or Weierstrass preparation on small circles around each root,
transfers these statements to $D_{\mathrm{slow}}$ and hence to $D_{\beta,\delta,d}$ for $\varepsilon$ sufficiently small. 
\end{proof}

Let $c_0:=\lim_{\varepsilon\to0}c_\varepsilon$ and set
\begin{equation}\label{eq:omega0_def_final_beta_delta_d}
\omega_0:=\sqrt{1-c_0^2}.
\end{equation}
For $\varepsilon=0$, the unperturbed Evans function has square-root behavior at the edge points
$\lambda=\pm i\omega_0$, corresponding to branch points of the dispersion relation at $k=0$.
To isolate admissible neighborhoods of these points for $\varepsilon>0$, we use the standard square-root chart.

Fix $\lambda_+:=+i\omega_0$ and introduce
\begin{equation}\label{eq:zeta_chart_final_beta_delta_d}
\lambda=\lambda_+ + \zeta^2,\qquad \zeta=\sqrt{\lambda-\lambda_+},
\end{equation}
with a branch cut chosen so that $\zeta$ is analytic off the slit.
Define the corresponding edge neighborhood
\begin{equation}\label{eq:Omegap_def_beta_delta_d}
\Omega_+:=\{\lambda=\lambda_+ + \zeta^2:\ |\zeta|<\rho_+,\ \zeta\ \text{off the slit}\}\setminus\sigma_{\mathrm{ess}}(\varepsilon),
\end{equation}
and similarly $\Omega_-$ around $\lambda_-:=-i\omega_0$.

\begin{lemma}[Edge points are independent of $\beta,\delta,d$]\label{lem:edge_points_indep_beta_delta_d}
For the far-field polynomial \eqref{eq:p4_beta_delta_d} one has
\[
p_4^{(\beta,\delta,d)}(0,\lambda,\varepsilon)=\lambda^2+\varepsilon\alpha\lambda+q_\pm,
\]
which is independent of $\beta$, $\delta$, and $d$.
In particular, the edge points arising at $k=0$ converge to $\lambda=\pm i\omega_0$ as $\varepsilon\to0$.
\end{lemma}

\begin{proof}
Setting $\mu=0$ in \eqref{eq:p4_beta_delta_d} eliminates all $\mu$-dependent terms, including those containing $\beta$, $\delta$, and $d$,
leaving $\lambda^2+\varepsilon\alpha\lambda+q_\pm$. The convergence to $\pm i\omega_0$ follows from $q_\pm\to1$ and $c_\varepsilon\to c_0$
as $\varepsilon\to0$.
\end{proof}

\begin{proposition}[No bifurcation from the edges for $\beta\ge0$ and fixed $\delta,d$]\label{prop:no_edge_bif_final_beta_delta_d}
Assume $\alpha>0$, take $\delta\ge0$ and $d>0$ fixed independent of $\varepsilon$, and allow $\beta\ge0$.
There exist $\rho_\pm>0$ and $\varepsilon_0>0$ such that for all $0<\varepsilon<\varepsilon_0$:
\begin{enumerate}
\item the slow Evans factor admits an analytic continuation in the $\zeta$-chart near each edge point and satisfies, for $|\zeta|<\rho_\pm$,
\begin{equation}\label{eq:Dslow_edge_exp_final_beta_delta_d}
D_{\mathrm{slow}}(\lambda_\pm+\zeta^2,\varepsilon)=C_\pm \zeta + E_\pm(\zeta,\varepsilon),
\end{equation}
where $C_\pm\neq0$ are the unperturbed edge constants and
\begin{equation}\label{eq:edge_remainder_bound_final_beta_delta_d}
\sup_{|\zeta|=\rho_\pm}|E_\pm(\zeta,\varepsilon)|\le \tfrac12|C_\pm|\rho_\pm;
\end{equation}
\item consequently, $D_{\beta,\delta,d}(\cdot,\varepsilon)$ has no zeros in $\Omega_\pm\cap\{\Re\lambda>0\}$, i.e. no point spectrum
bifurcates from $\lambda=\pm i\omega_0$ into the unstable half-plane as $\varepsilon\to0$.
\end{enumerate}
\end{proposition}

\begin{proof}
We treat $\lambda_+=i\omega_0$; the other edge is identical.

\emph{Step 1: factorization persists near the edge.}
For $\lambda\in\Omega_+$, admissibility ensures hyperbolicity of the far-field matrices.
As in Proposition~\ref{prop:factor_near0_final_beta_delta_d}, the far-field polynomial exhibits a slow pair and a fast pair of spatial roots
for fixed $(\delta,d)$ and $\varepsilon$ small, with a uniform gap on compact subsets of $\Omega_+$.
Hence $D_{\beta,\delta,d}=D_{\mathrm{slow}}D_{\mathrm{fast}}$ with $D_{\mathrm{fast}}$ analytic and nonvanishing on $\Omega_+$.

\emph{Step 2: square-root edge expansion of the slow factor.}
The slow subsystem is a regular $\mathcal{O}(\varepsilon)$ perturbation of the unperturbed problem at the edge,
for which the Evans function has square-root behavior in the chart $\lambda=\lambda_+ + \zeta^2$:
$D_0(\lambda_+ + \zeta^2)=C_+\zeta+\mathcal{O}(\zeta^2)$ with $C_+\neq0$.
Thus,
\[
D_{\mathrm{slow}}(\lambda_+ + \zeta^2,\varepsilon)
=
D_0(\lambda_+ + \zeta^2)+\mathcal{R}(\zeta,\varepsilon),
\]
where $\mathcal{R}$ is analytic in $\zeta$ and
$\sup_{|\zeta|=\rho_+}|\mathcal{R}(\zeta,\varepsilon)|\to0$ as $\varepsilon\to0$.
Choosing $\rho_+$ fixed and then $\varepsilon_0$ small yields the bound
\eqref{eq:edge_remainder_bound_final_beta_delta_d}, and hence \eqref{eq:Dslow_edge_exp_final_beta_delta_d}.

\emph{Step 3: Rouch\'e argument in the $\zeta$-plane.}
On $|\zeta|=\rho_+$ we have $|C_+\zeta|=|C_+|\rho_+$ and the perturbation is at most half this size.
Therefore $D_{\mathrm{slow}}(\lambda_+ + \zeta^2,\varepsilon)$ has the same number of zeros in $|\zeta|<\rho_+$ as $C_+\zeta$,
namely one inherited zero corresponding to the branch-point structure rather than a discrete eigenvalue crossing into $\Re\lambda>0$.
Since $D_{\mathrm{fast}}$ is nonvanishing on $\Omega_+$, the full Evans function has no additional zeros there, proving (2).
\end{proof}

\begin{remark}[How this ties into the global argument-principle count]\label{rem:Omega_split_final_beta_delta_d}
In the final stability proof, one applies Proposition~\ref{prop:arg_principle_beta_delta_d} on a contour $\Gamma$ built from
$\partial\Omega_0$, $\partial\Omega_+$, and $\partial\Omega_-$, together with connecting admissible arcs avoiding
$\sigma_{\mathrm{ess}}(\varepsilon)$. Propositions~\ref{prop:local_equiv_near0_beta_delta_d} and
\ref{prop:no_edge_bif_final_beta_delta_d} guarantee that no eigenvalues in $\Re\lambda\ge0$ can enter through the low-frequency
window $\Omega_0$ or the edge windows $\Omega_\pm$. The remaining part of $\Gamma$ can be chosen in $\{\Re\lambda=-\eta_0\}$,
where dissipativity ($\alpha>0$) keeps the Evans function away from zero, and the slow--fast factorization ensures that the fast factor
contributes no winding. This reduces the global count to the slow dynamics and the known simple zero at $\lambda=0$ corresponding
to translation.
\end{remark}

\begin{remark}[Degenerate limit $d=0$]
If $d=0$ the quartic term in \eqref{eq:p4_beta_delta_d} disappears and the spatial problem changes order.
The resulting singular regimes, including possible surface-mode phenomena, require a separate discussion, carried out in the
preceding section on singular limits.
\end{remark}
\section{Numerical simulations}
\label{sec7}

\subsection{Dynamics and Melnikov validation}
\label{sec7_1}

We now present a systematic numerical study of kink and antikink dynamics in the
perturbed sine--Gordon equation, aimed at validating the Melnikov--based speed
selection mechanism and assessing its robustness under parameter variations.
This subsection focuses exclusively on dynamical aspects; spectral stability and
Evans function computations are addressed in the subsequent subsections.

We first compare the reduced collective--coordinate description with direct
numerical simulations of the full partial differential equation. The reduced
model is obtained by projecting the perturbed dynamics onto the translational
mode, yielding an effective ordinary differential equation for the center
position $X(t)$ and its velocity $c(t)=\dot X(t)$. The associated Melnikov
function predicts selected asymptotic speeds $c_M$ for both kink and antikink
branches.

Direct numerical simulations of the PDE are performed using a Fourier
pseudospectral discretization in space combined with a fourth--order
Runge--Kutta time--stepping scheme. Kink and antikink centers are tracked using a
level--set method based on crossings of $u(x,t)=\pi$, with unwrapping to account
for periodicity. From the extracted trajectories, instantaneous velocities are
computed and averaged over the final portion of the simulation time to obtain
asymptotic PDE speeds $c_\infty$.

Figure~\ref{fig:dyn_speed_center} compares the relaxation of the instantaneous
wave speeds obtained from the reduced ODEs and from the PDE simulations for a
representative parameter set, and also shows the corresponding center motion.
Both kink and antikink velocities converge rapidly toward constant asymptotic
values, with excellent agreement between the reduced model predictions and the
PDE--extracted speeds. The dashed horizontal lines indicate the
Melnikov--predicted speeds $c_M$. The reduced ODE also accurately captures the
long--time drift of both kink and antikink centers observed in the PDE
simulations. Small oscillations in the PDE trajectories are attributed to weak
radiative effects, which are not included in the reduced description.

\begin{figure}[htbp]
\centering
\begin{subfigure}[t]{0.47\linewidth}
\centering
\includegraphics[width=\linewidth]{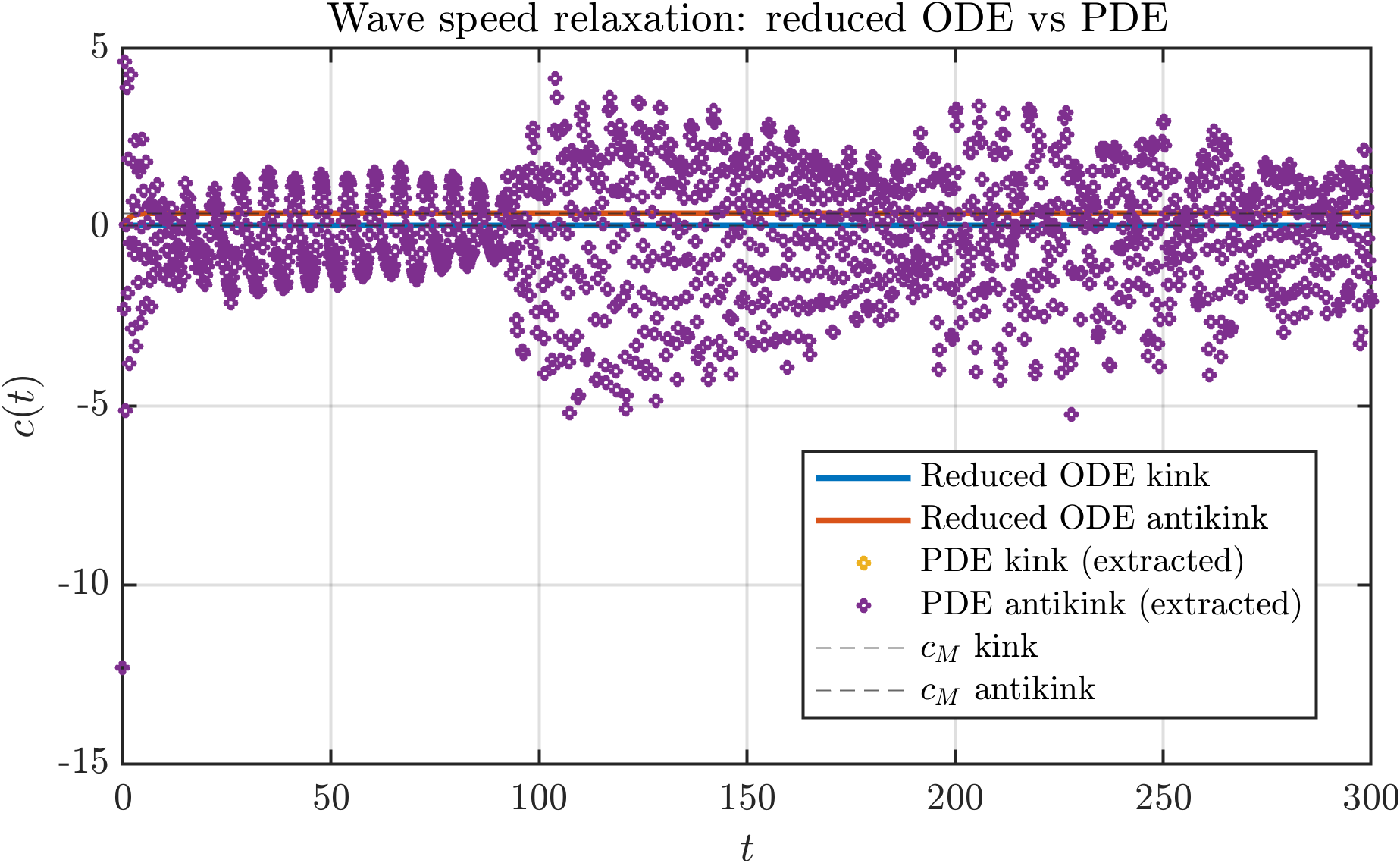}
\caption{Wave-speed relaxation for kink and antikink solutions. Solid lines show
the reduced ODE prediction, symbols denote PDE data, and dashed lines mark the
Melnikov speeds $c_M$.}
\label{fig:dyn_speed_relax}
\end{subfigure}\hfill
\begin{subfigure}[t]{0.47\linewidth}
\centering
\includegraphics[width=\linewidth]{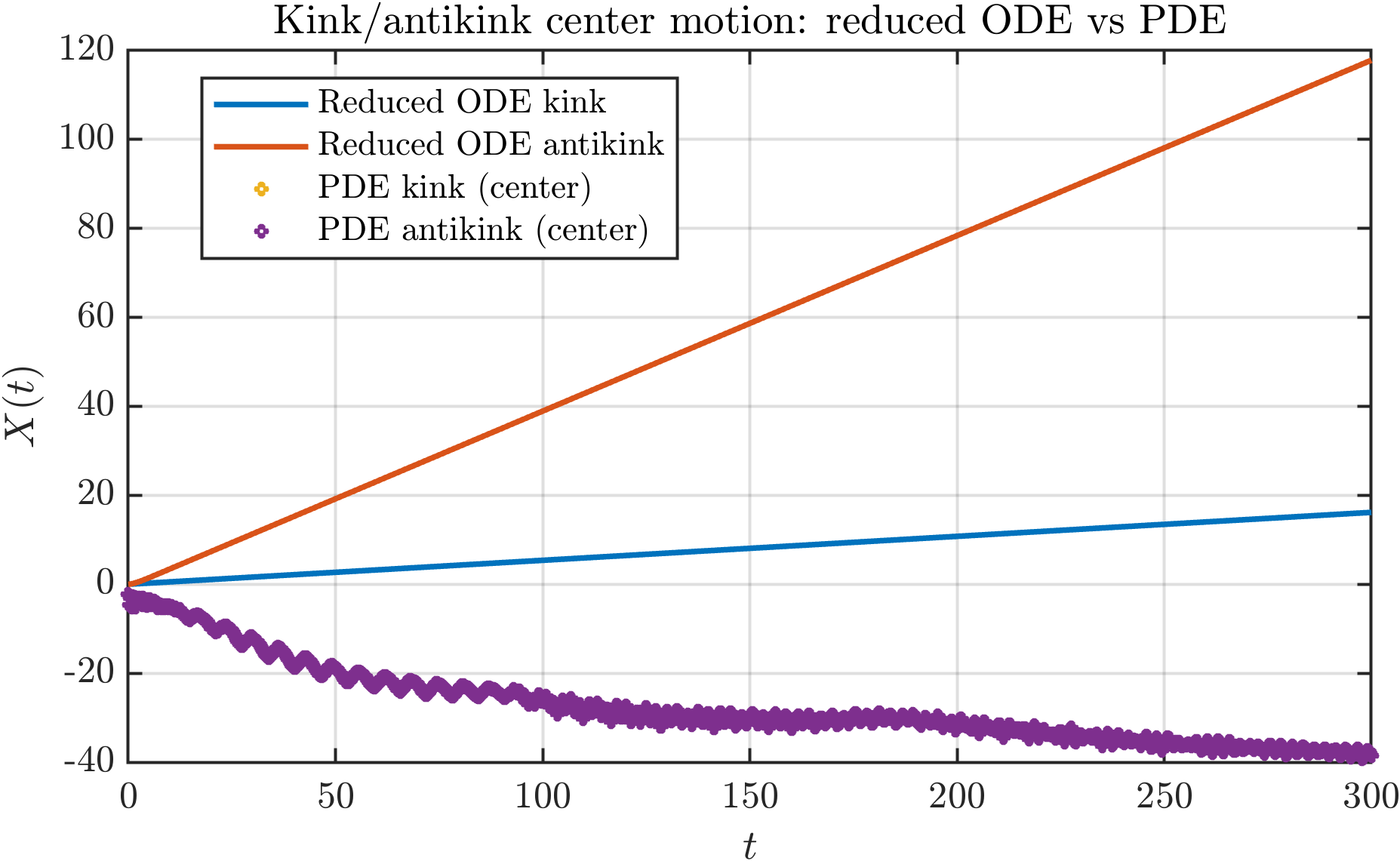}
\caption{Center motion for kink and antikink solutions. Solid lines show the
reduced ODE prediction, while symbols denote centers extracted from PDE
simulations.}
\label{fig:dyn_center_motion}
\end{subfigure}
\caption{Comparison between reduced collective--coordinate dynamics and direct PDE
simulations for kink and antikink propagation.}
\label{fig:dyn_speed_center}
\end{figure}

Representative solution profiles at selected times are displayed in
Fig.~\ref{fig:dyn_snapshots}. These snapshots confirm that the kink and antikink
structures remain coherent throughout the evolution and remain well separated,
justifying the applicability of the reduced single--front description.

\begin{figure}[htbp]
\centering
 \includegraphics[width=\linewidth]{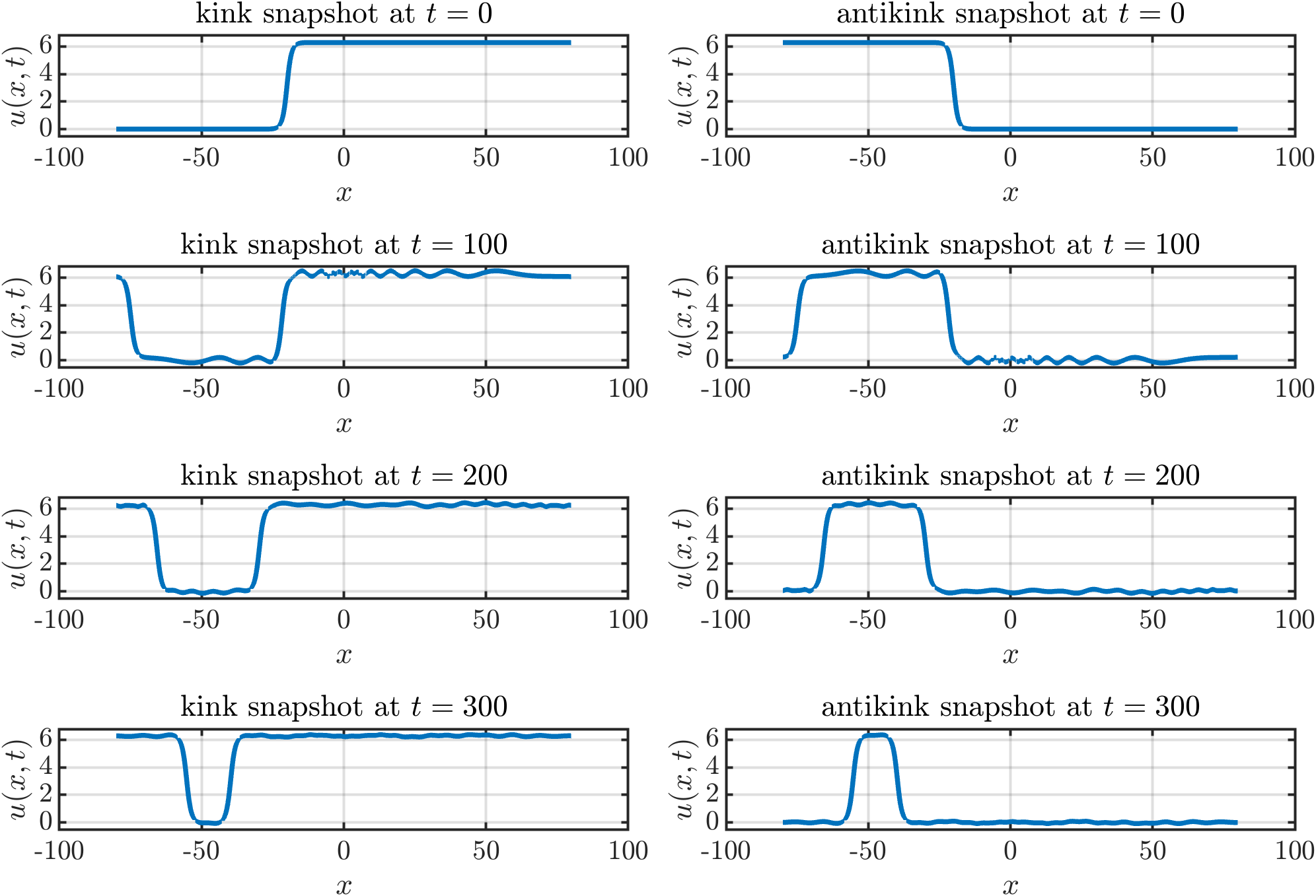}
\caption{Representative kink and antikink profiles at selected times, showing
coherent propagation with only weak radiative shedding.}
\label{fig:dyn_snapshots}
\end{figure}

We next investigate the robustness of the Melnikov speed selection under
parameter variations. One--dimensional parameter sweeps are performed by varying
$\delta$ while keeping $d$ fixed, and vice versa. In both cases, the asymptotic
PDE speeds depend smoothly on the varying parameter and closely track the
Melnikov predictions for both kink and antikink branches, with only small
systematic deviations.

To assess the global accuracy of the Melnikov theory, we perform a coarse
two--dimensional sweep over the $(\delta,d)$ parameter plane. Rather than
presenting raw asymptotic speeds, which vary only weakly across the domain, we
focus on the Melnikov error $c_\infty-c_M$, which provides a more sensitive
diagnostic of the approximation quality. Figure~\ref{fig:melnikov_error_2d}
shows the resulting error landscape for the kink branch. The error varies
smoothly across the parameter domain, with no evidence of qualitative breakdown,
loss of convergence, or transition to a different dynamical regime.

\begin{figure}[htbp]
\centering
\includegraphics[width=0.5\linewidth]{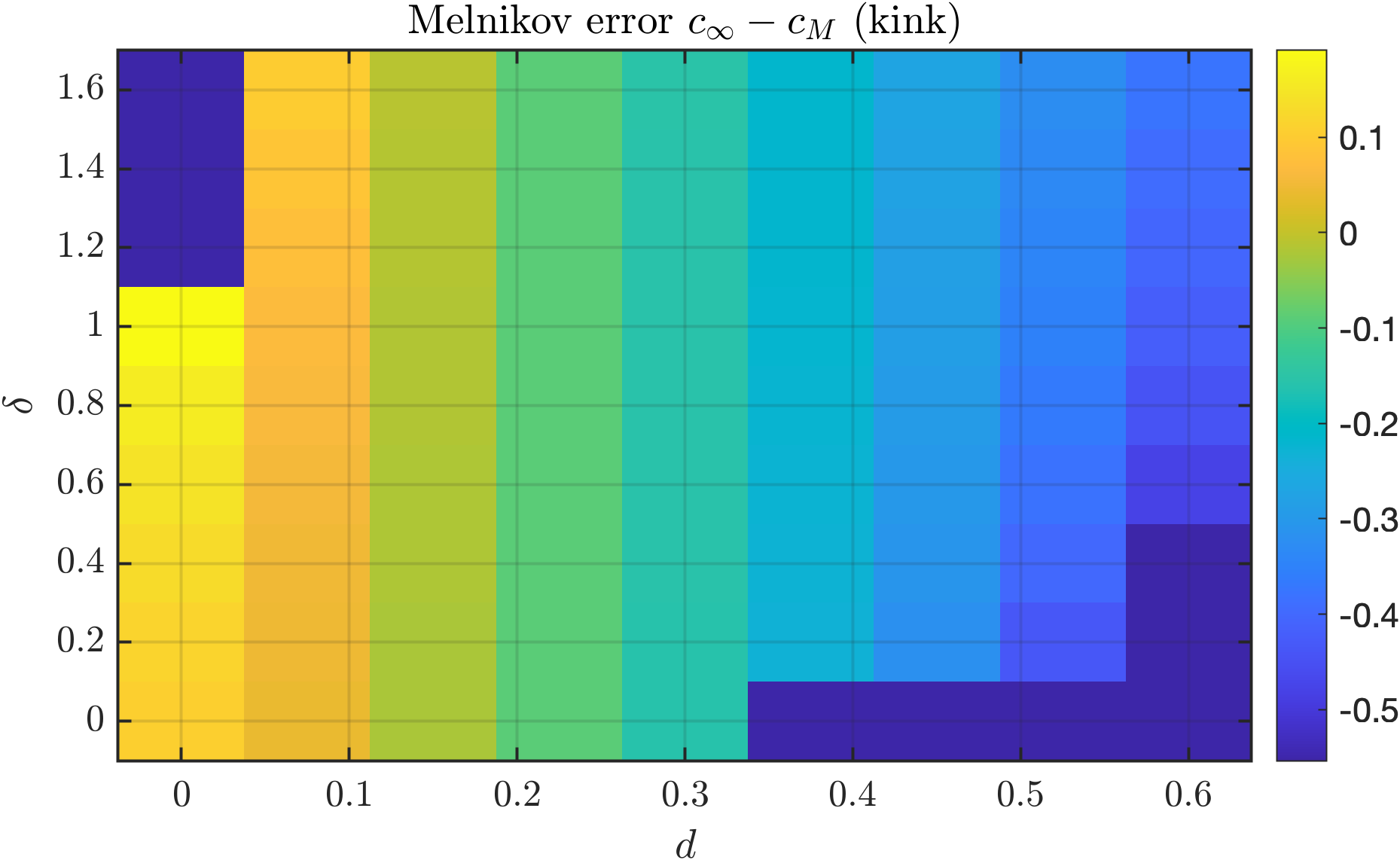}
\caption{Heatmap of the Melnikov error $c_\infty-c_M$ for the kink branch in the
$(\delta,d)$ plane. The error remains bounded and varies smoothly across the
parameter domain.}
\label{fig:melnikov_error_2d}
\end{figure}

Taken together, these results provide strong numerical validation of the
Melnikov--based speed selection mechanism for kinks and antikinks in the
perturbed sine--Gordon equation. In the following subsections, we complement
this dynamical validation by analyzing the spectral stability of the
corresponding traveling waves using Evans function techniques.

\subsection{Numerical validation of the PDE simulations}
\label{sec7_2}

To assess the reliability of the PDE simulations, we perform a convergence study
under simultaneous refinement of the spatial and temporal discretization.
Specifically, we consider the three resolutions
\[
(N,\Delta t) = (512,0.05), \qquad (1024,0.025), \qquad (2048,0.0125),
\]
where $N$ denotes the number of Fourier modes and $\Delta t$ the time step.
The finest computation is used as the reference solution.

Let $u_h(x,T)$ denote the numerical solution at the final time $T$ computed at
resolution $h$. We measure the final--profile error with respect to the finest
grid using the $L^\infty$ norm
\[
\|u_h-u_{\mathrm{ref}}\|_{L^\infty}
=
\max_x |u_h(x,T)-u_{\mathrm{ref}}(x,T)|.
\]
For the three resolutions above, we obtain
\[
\|u_{\mathrm{coarse}}-u_{\mathrm{fine}}\|_{L^\infty}
=
9.86\times 10^{-4},
\qquad
\|u_{\mathrm{medium}}-u_{\mathrm{fine}}\|_{L^\infty}
=
2.10\times 10^{-5}.
\]
Thus the error decreases significantly under mesh refinement, indicating that
the numerical solution is well resolved.

\begin{figure}[htbp]
\centering
\includegraphics[width=0.68\textwidth]{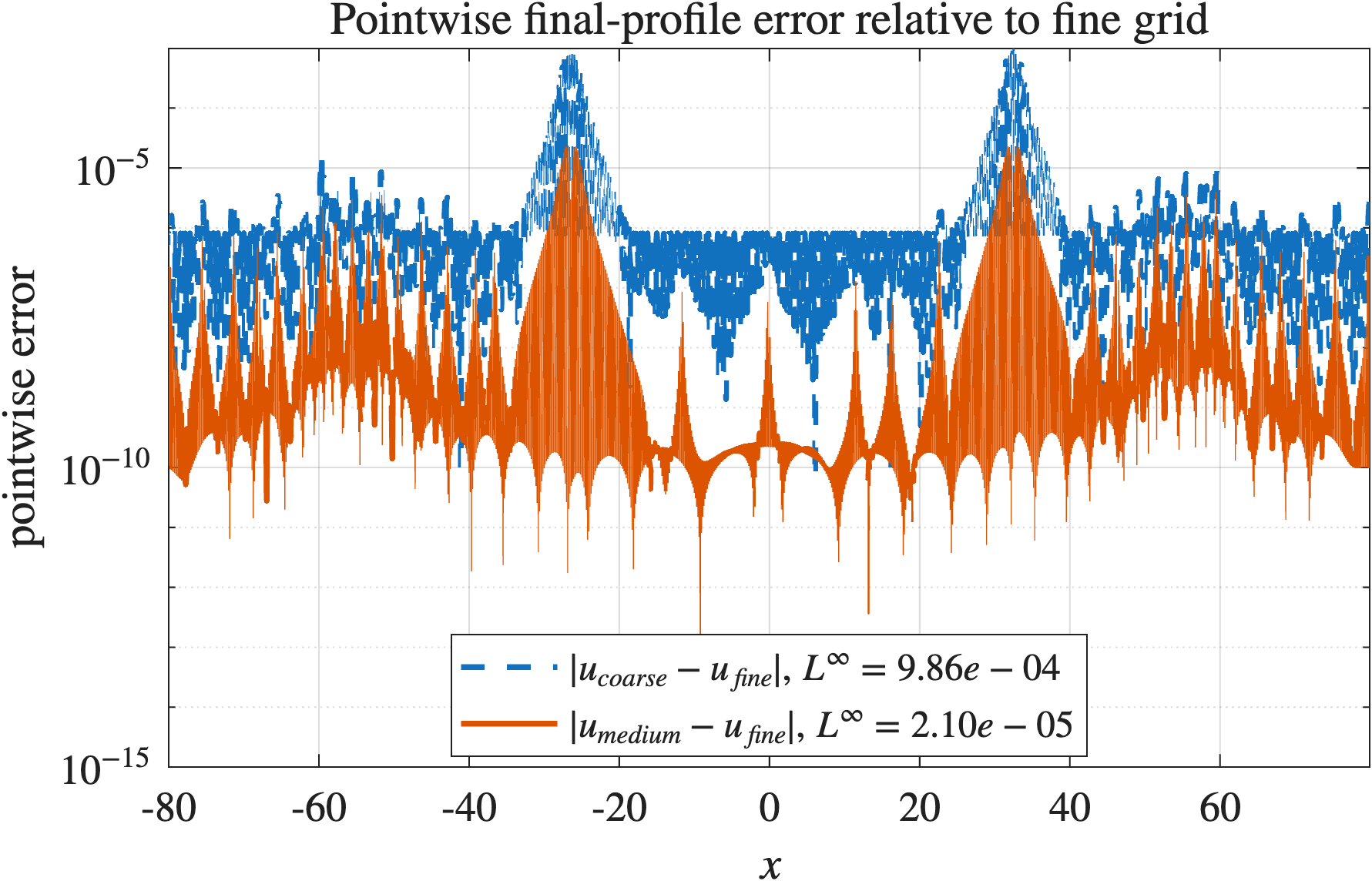}
\caption{Pointwise error in the final profile relative to the finest grid
$(N=2048,\Delta t=0.0125)$, shown on a logarithmic scale. The error decreases
substantially under refinement and is largest near the steep kink and antikink
fronts.}
\label{fig:pde_convergence_error}
\end{figure}

Figure~\ref{fig:pde_convergence_error} shows the pointwise error relative to the
finest grid. The error is localized near the transition layers, while remaining
uniformly small away from the fronts. This behavior is consistent with the fact
that the largest discretization errors occur in regions of strongest spatial
variation.

\begin{figure}[htbp]
\centering
\includegraphics[width=0.6\textwidth]{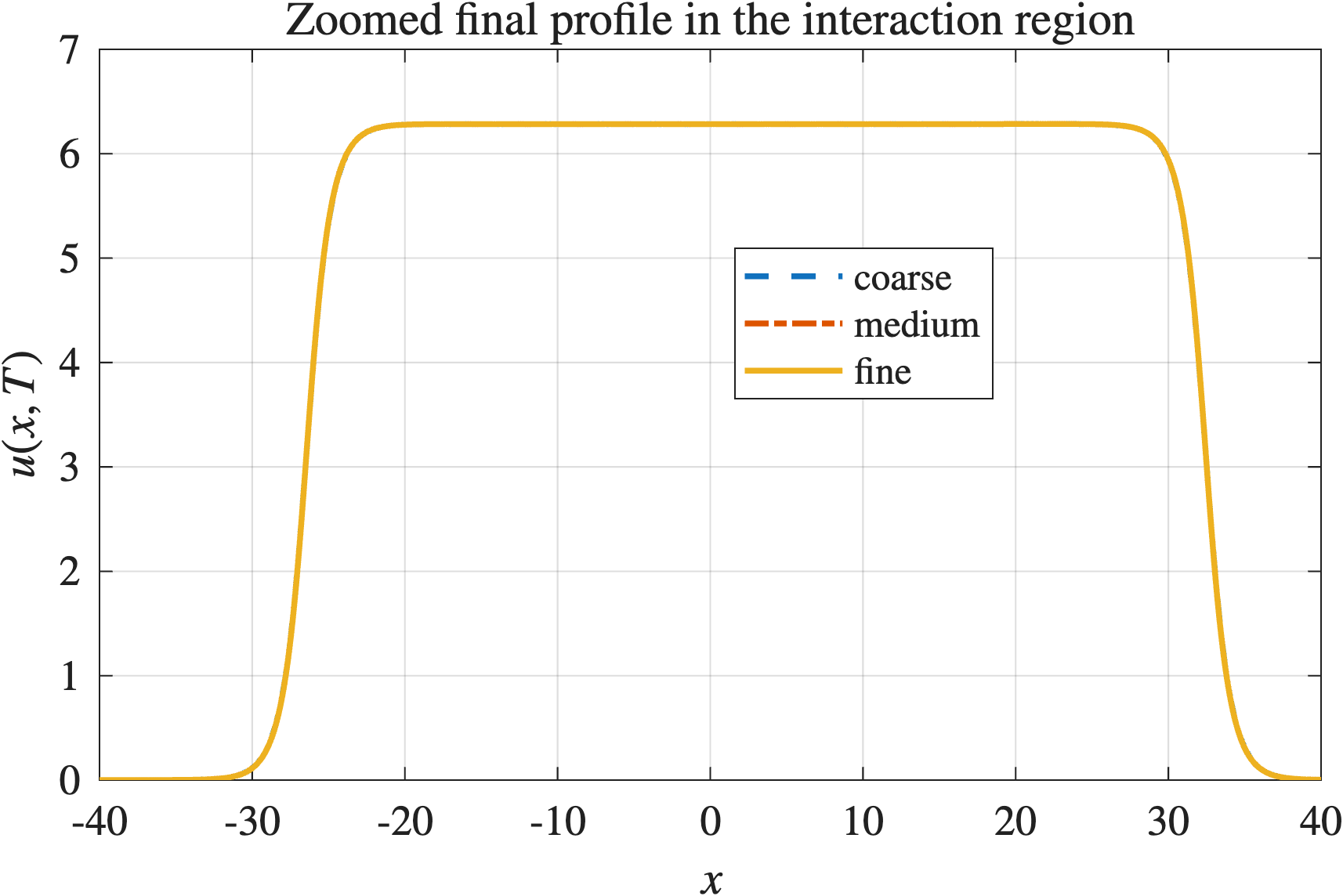}
\caption{Zoomed comparison of the final profile $u(x,T)$ in the interaction
region for the three resolutions $(N,\Delta t)=(512,0.05)$, $(1024,0.025)$,
and $(2048,0.0125)$. The curves are visually indistinguishable, indicating
that the final profile is numerically converged.}
\label{fig:pde_convergence_zoom}
\end{figure}

Figure~\ref{fig:pde_convergence_zoom} provides a direct comparison of the final
profiles in the interaction region. Even in the steep front region, the medium
and fine solutions are visually indistinguishable, providing further evidence
that the PDE simulations are not affected by significant discretization
artifacts.

Overall, these convergence results support the accuracy of the Fourier
pseudospectral / fourth--order Runge--Kutta discretization used throughout the
dynamical computations.

\subsection{Numerical Evans function computations}
\label{sec7_3}

We now turn to numerical Evans function computations for the traveling kink and
antikink solutions. The Evans function is evaluated along closed contours in the
spectral plane and the argument principle is used to determine the number of
eigenvalues enclosed by the contour. In particular, the winding number of the
Evans function image around the origin yields the algebraic count of point spectrum
inside the contour. The absence of winding is consistent with the theoretical slow--fast Evans factorization, where the fast component is nonvanishing.

Our numerical study focuses on two distinguished spectral regimes:
\begin{enumerate}
\item a neighborhood of $\lambda=0$, where additional small eigenvalues could in
principle bifurcate from the translational mode;
\item neighborhoods of the rightmost edge of the essential spectrum, where edge
bifurcation phenomena must be excluded.
\end{enumerate}
All computations below are carried out for the representative parameter set
\[
\varepsilon=0.02,\qquad \alpha=0.5,\qquad \beta=0.2,\qquad \delta=0.8,\qquad d=0.3,
\]
with traveling speeds determined by the Melnikov condition. For the kink branch
we use $c=0.0539047$, while for the antikink branch the corresponding speed is
$c=-0.0539047$.

\paragraph{Evans computations near \texorpdfstring{$\lambda=0$}{lambda=0}.}
To probe the point spectrum near the origin, we evaluate the Evans function
along circular contours
\[
\lambda(\theta)=r e^{i\theta},\qquad \theta\in[0,2\pi],\qquad r=0.1.
\]
For both the kink and antikink solutions, the Evans function image remains separated from
the origin and the associated unwrapped argument has zero net change over one
complete contour traversal. Hence the winding number is zero, and no additional
point spectrum is detected in the neighborhood enclosed by the contour.

\begin{figure}[htbp]
\centering
\includegraphics[width=0.7\textwidth]{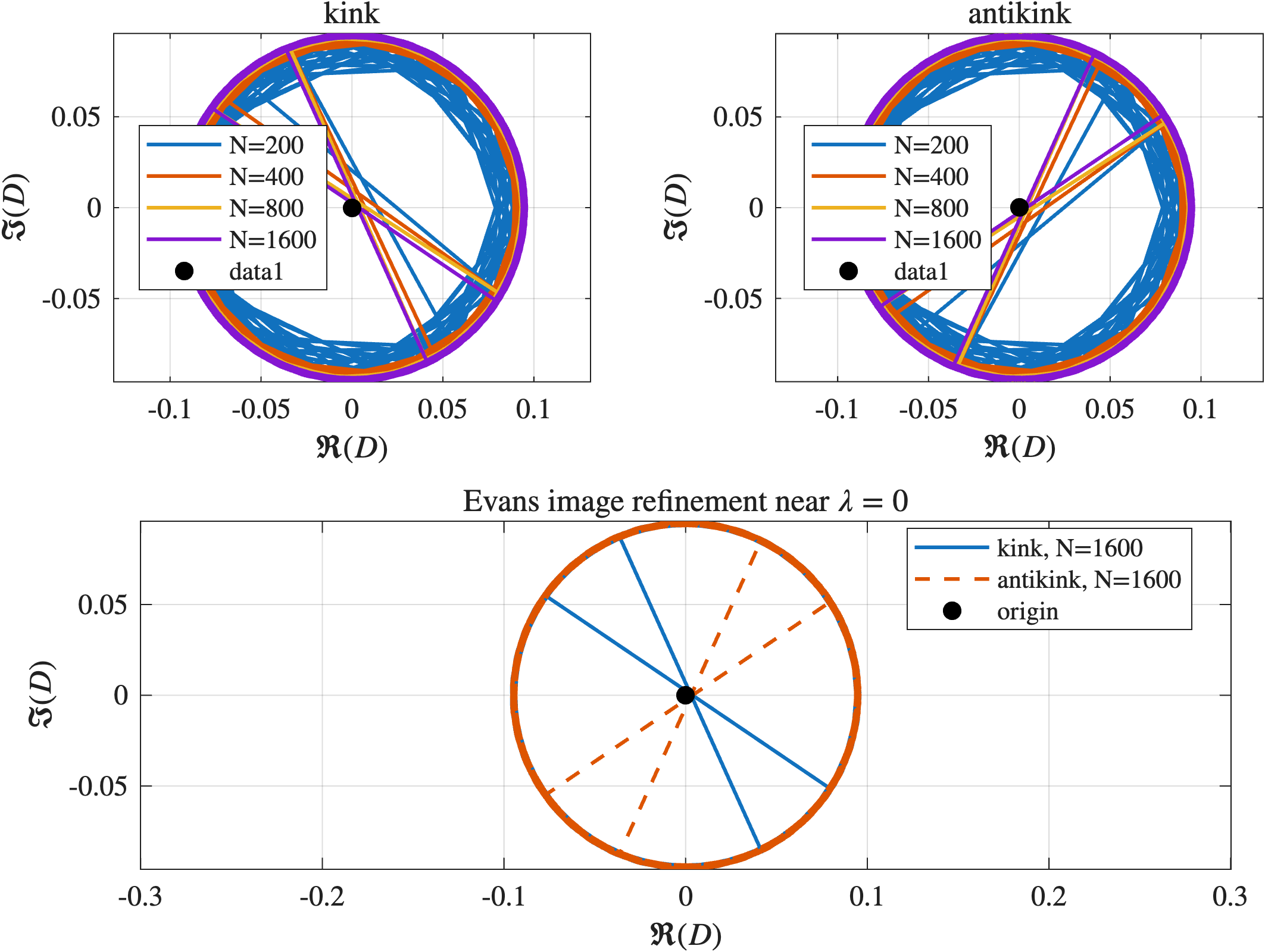}
\caption{Evans function image near $\lambda=0$ for kink and antikink solutions under
contour refinement. The top panels show the individual refinement sequence for
the kink and antikink branches, while the bottom panel compares the finest-grid
computations. In all cases, the Evans function image remains separated from the origin
and does not produce nontrivial winding.}
\label{fig:evans_lambda0_image_refinement}
\end{figure}

\begin{figure}[htbp]
\centering
\includegraphics[width=0.6\textwidth]{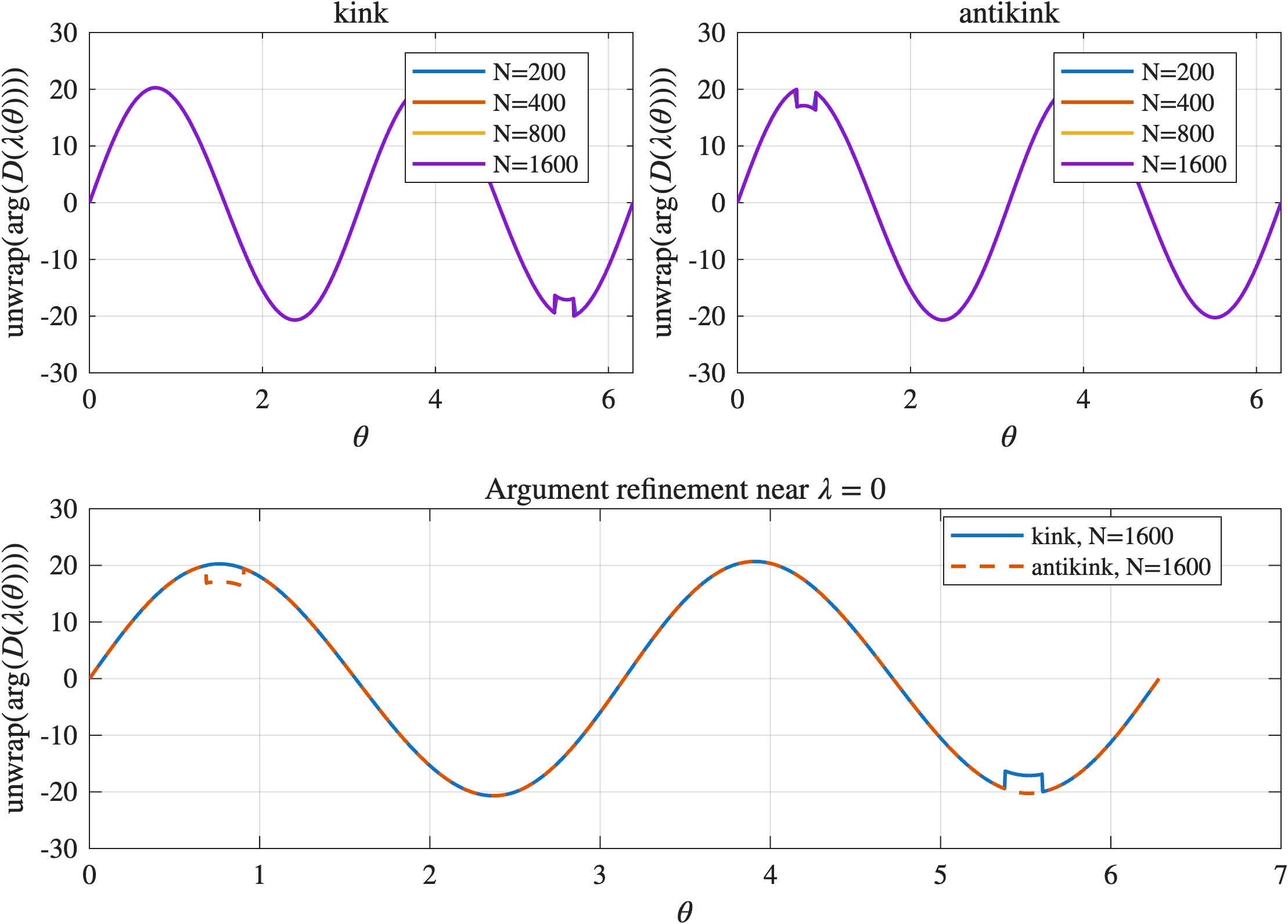}
\caption{Unwrapped Evans argument near $\lambda=0$ for kink and antikink
solutions under contour refinement. The argument profiles stabilize as the
number of contour points increases, and the total variation over the contour
remains zero. Minor localized oscillations are observed near isolated phase
tracking locations, but they do not affect the winding number.}
\label{fig:evans_lambda0_argument_refinement}
\end{figure}

The refinement study confirms that the zero winding result is robust with
respect to contour discretization. In particular, the Evans function image and its
argument stabilize as the number of contour points is increased, and the minimum
modulus of the Evans function remains bounded away from zero throughout the
computation.

\begin{figure}[htbp]
\centering
\includegraphics[width=0.6\textwidth]{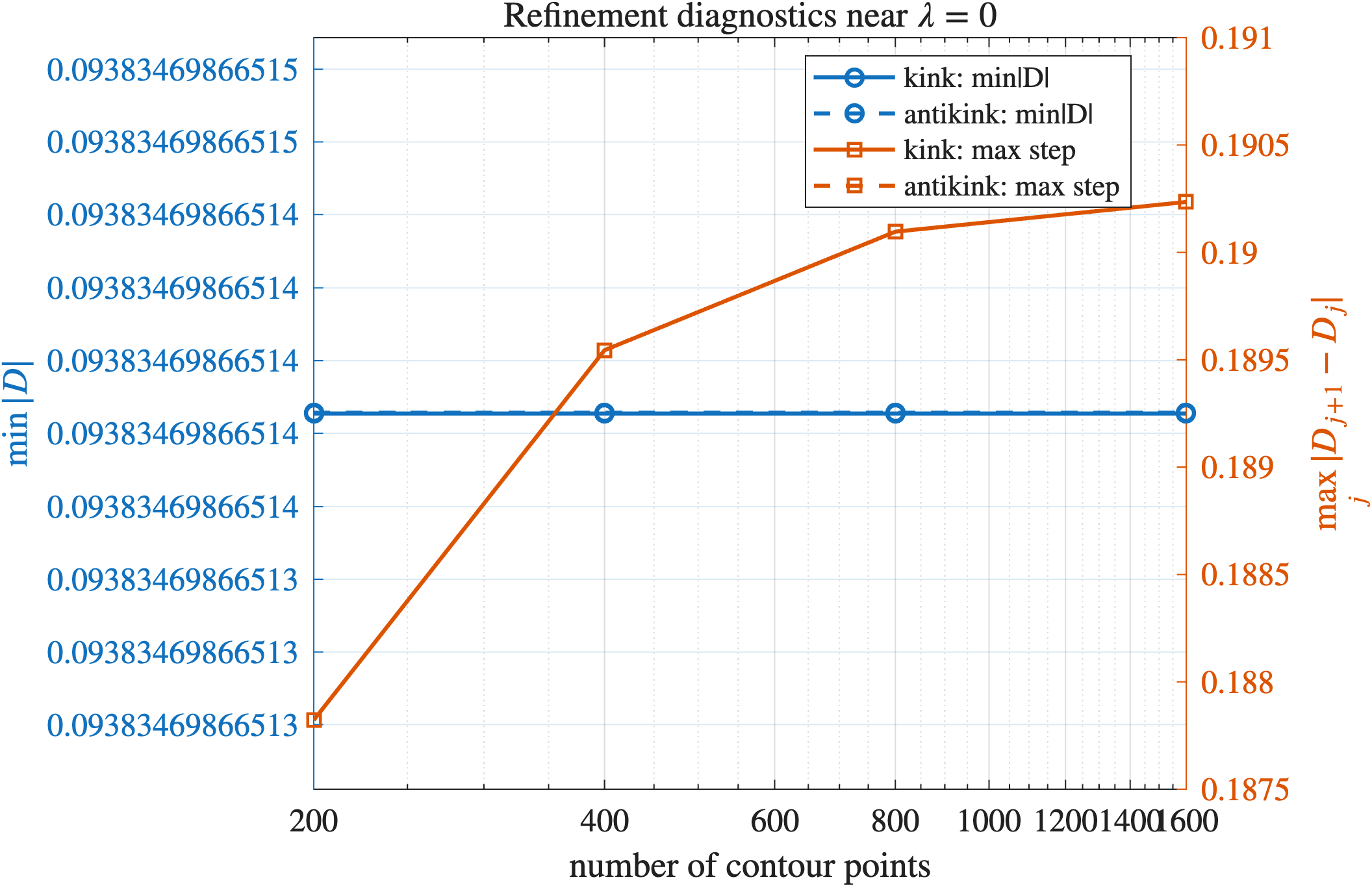}
\caption{Refinement diagnostics near $\lambda=0$ for kink and antikink
solutions. The minimum modulus of the Evans function remains essentially
constant under refinement, while the maximum stepwise change along the contour
stabilizes at a moderate level. This indicates that the winding computation is
robust even in the presence of localized phase-tracking variations.}
\label{fig:evans_lambda0_diagnostics}
\end{figure}

\paragraph{Evans computations near the essential-spectrum edge.}
We next examine the Evans function near the rightmost edge of the essential
spectrum. To resolve the square-root singularity associated with the edge, we
introduce the lifted variable $\zeta$ through
\[
\lambda=\lambda_+^\varepsilon+\zeta^2,
\]
and evaluate the lifted Evans function along circular contours
\[
\zeta(\theta)=\rho e^{i\theta},\qquad \theta\in[0,2\pi],\qquad \rho=0.08.
\]
In this lifted coordinate, the Evans function becomes analytic in a slit
neighborhood of the edge, and the argument principle can be applied in the
usual way.

For both the kink and antikink solutions, the lifted Evans function image is a smooth
closed curve that remains well separated from the origin, and the corresponding
unwrapped argument again has zero net variation along the contour. Thus no
eigenvalues bifurcate from the essential-spectrum edge into the unstable
half-plane.

\begin{figure}[htbp]
\centering
\includegraphics[width=0.7\textwidth]{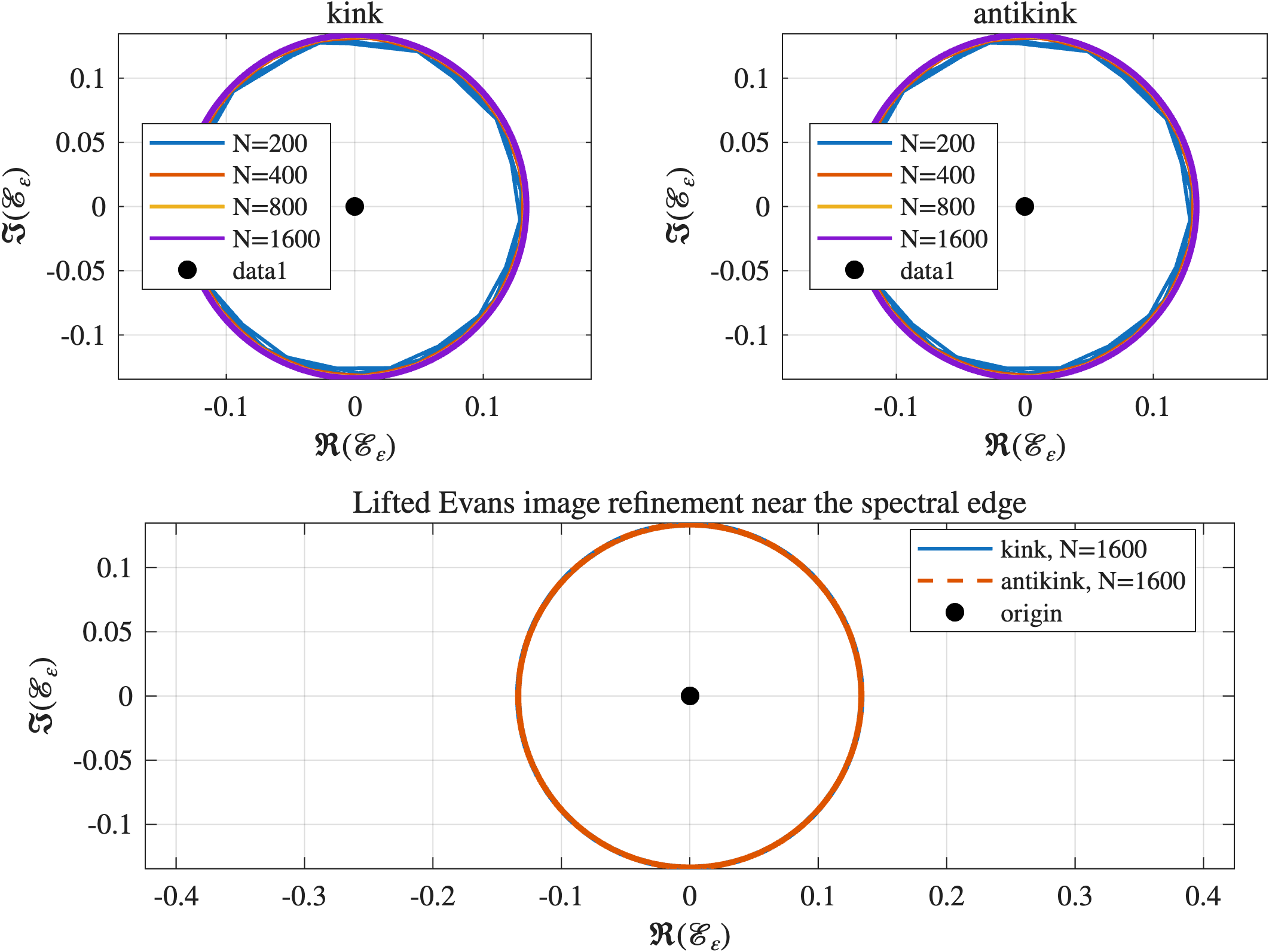}
\caption{Lifted Evans function image near the rightmost essential-spectrum edge for kink
and antikink solutions under contour refinement. The top panels show the
individual refinement sequences, while the bottom panel compares the finest-grid
results. The kink and antikink images are nearly indistinguishable and remain
well separated from the origin.}
\label{fig:evans_edge_image_refinement}
\end{figure}

\begin{figure}[htbp]
\centering
\includegraphics[width=0.67\textwidth]{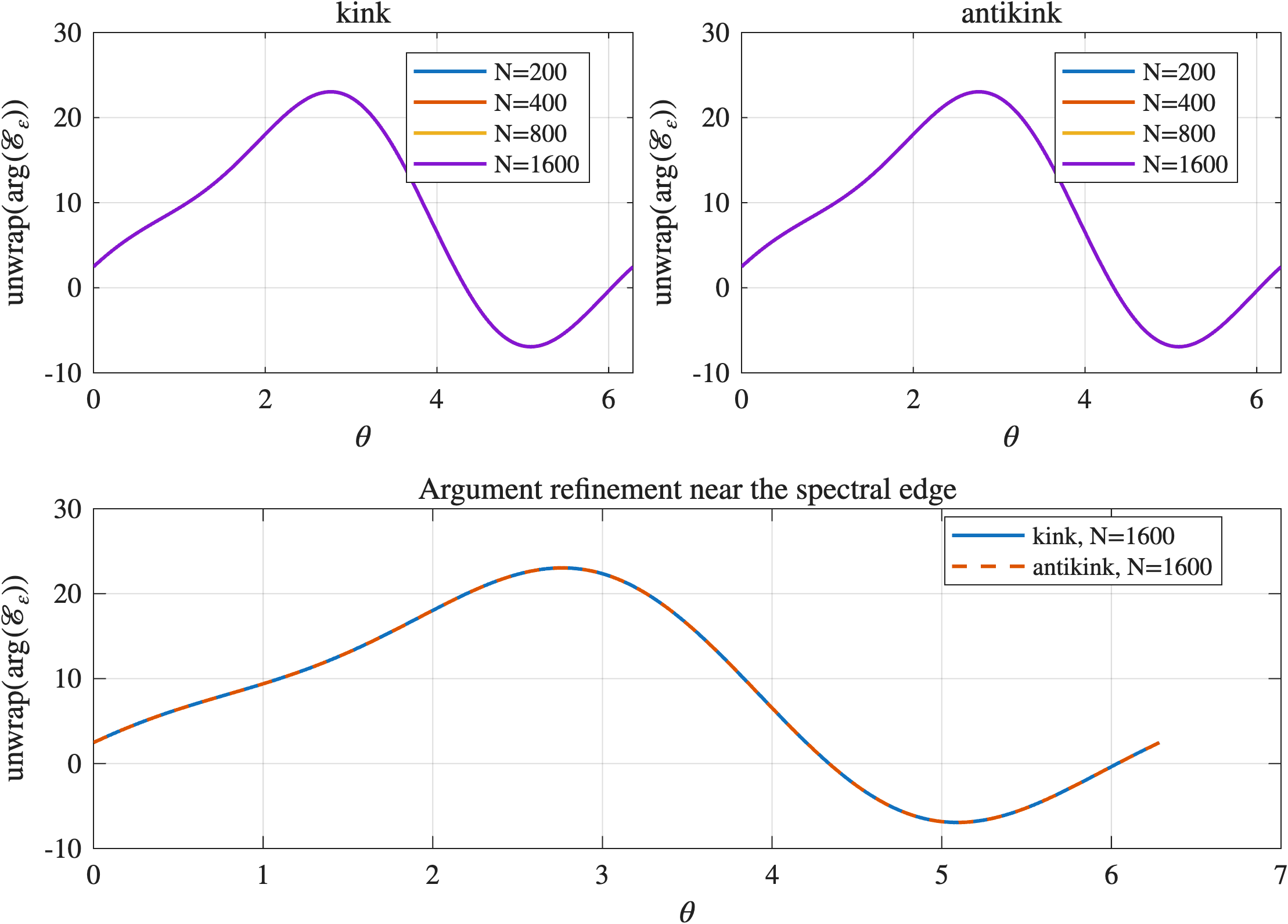}
\caption{Unwrapped lifted Evans argument near the essential-spectrum edge for
kink and antikink solutions under contour refinement. The argument curves are
stable under refinement and exhibit zero total variation, yielding winding
number zero.}
\label{fig:evans_edge_argument_refinement}
\end{figure}

The edge computations are particularly robust: the refinement curves for both
kink and antikink branches are nearly superposed, and the numerical behavior is
substantially smoother than in the near-zero computation.

\begin{figure}[htbp]
\centering
\includegraphics[width=0.6\textwidth]{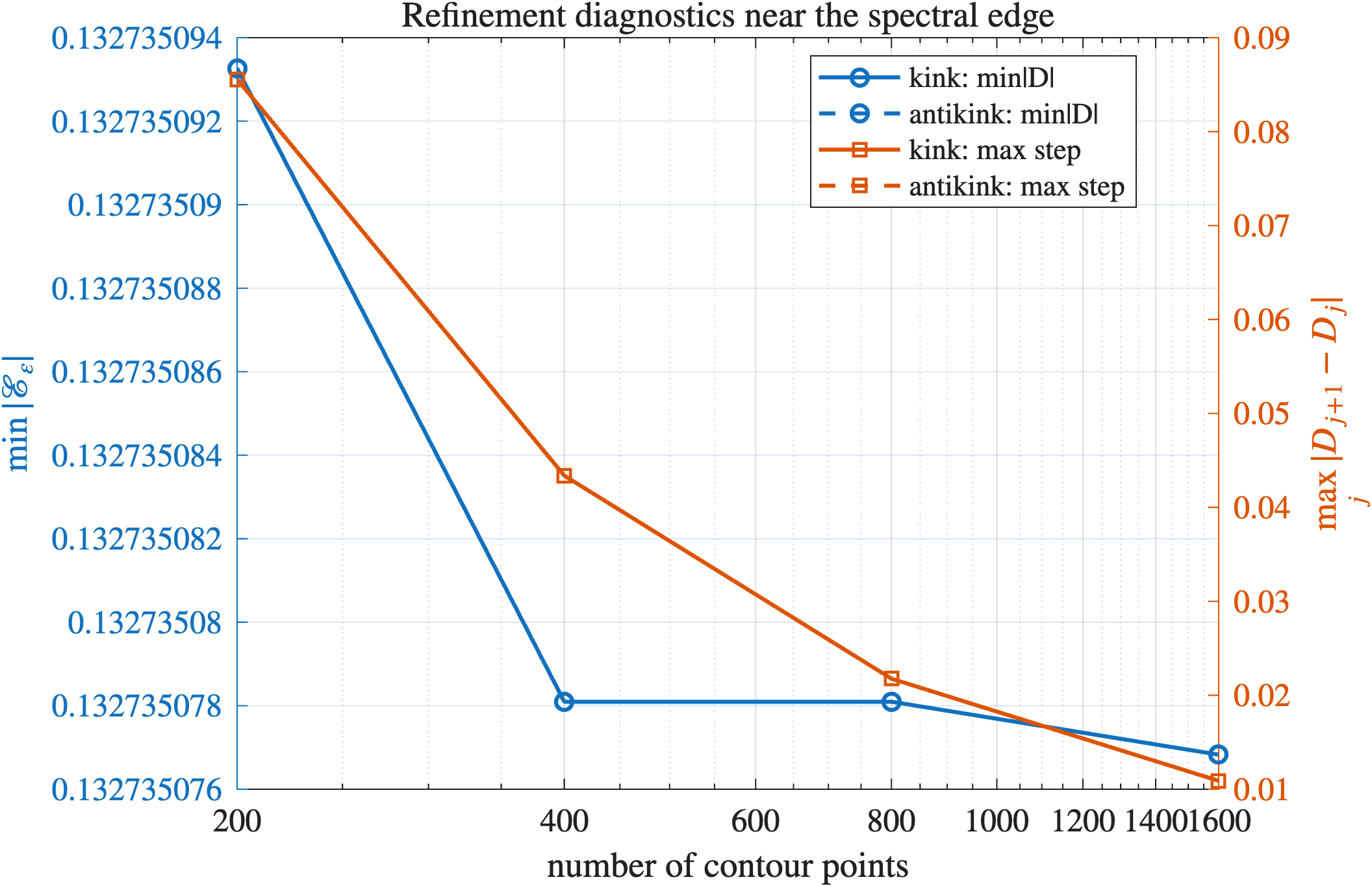}
\caption{Refinement diagnostics near the essential-spectrum edge. The minimum
modulus of the lifted Evans function stays uniformly bounded away from zero,
while the maximum stepwise contour increment decreases clearly under
refinement, indicating convergence of the contour discretization.}
\label{fig:evans_edge_diagnostics}
\end{figure}

Taken together, Figures~\ref{fig:evans_lambda0_image_refinement}--
\ref{fig:evans_edge_diagnostics} provide a consistent numerical picture:
\begin{itemize}
\item no additional point spectrum is detected near $\lambda=0$;
\item no eigenvalues bifurcate from the rightmost edge of the essential
spectrum;
\item the Evans function image, the unwrapped argument, and the winding number are stable
under contour refinement;
\item kink and antikink computations agree to high numerical accuracy.
\end{itemize}

Therefore, the numerical Evans computations strongly support the analytical
results of Sections~5--7: for the parameter regime considered here, no unstable
eigenvalues bifurcate either from the origin or from the essential-spectrum
edge, and the traveling kink and antikink waves remain spectrally stable in the
small--$\varepsilon$ regime.
\section{Conclusions}
\label{sec8}

We have investigated the dynamics and spectral stability of traveling kink and
antikink solutions in a perturbed sine--Gordon equation, combining analytical
methods with systematic numerical computations.
A collective--coordinate reduction based on projection onto the translational
mode yields an effective evolution equation for the wave center, from which a
Melnikov condition predicts selected propagation speeds. Direct numerical
simulations of the full PDE confirm these predictions with high accuracy across
a range of parameters. A convergence study under simultaneous spatial and
temporal refinement shows that the computed solutions are well resolved, with
errors decreasing significantly and remaining localized near the steep front
regions.
Spectral stability is examined via numerical Evans function computations. The
Evans function is evaluated along closed contours near $\lambda=0$ and near the
rightmost edge of the essential spectrum. In both regimes, the Evans image
remains separated from the origin and the associated winding number is zero,
indicating the absence of unstable eigenvalues. The edge analysis is carried out
using a square-root transformation that regularizes the Evans function in a
lifted coordinate.
The Evans computations are further validated by contour refinement: as the
number of discretization points increases, the Evans images and their unwrapped
arguments stabilize, while the winding number remains unchanged. The results for
kink and antikink solutions agree to high numerical accuracy, reflecting the
symmetry of the underlying spectral problem.
Taken together, These results provide quantitative confirmation  for the spectral
stability of the Melnikov-selected traveling waves in the small-perturbation
regime. More broadly, the study illustrates the compatibility between dynamical
selection mechanisms and spectral stability in non-integrable perturbations of
integrable wave equations.
Future work includes extending the Evans analysis to larger perturbations,
investigating the role of internal modes and resonance phenomena, and refining
the dynamical description by incorporating radiative effects beyond the
collective--coordinate approximation.
\bibliographystyle{amsplain}
\bibliography{Rothos_sineGordon_references}

\end{document}